\documentclass[%
 reprint,
 amsmath,amssymb,
 aps,
pra,
showkeys,
]{revtex4-2}

\usepackage{graphicx}
\usepackage{dcolumn}
\usepackage{bm}
\usepackage{hyperref}
\hypersetup{
  colorlinks   = true, 
  urlcolor     = blue, 
  linkcolor    = blue, 
  citecolor    = blue  
}

\usepackage{amsthm}
\usepackage{braket}
\usepackage[dvipsnames]{xcolor}
\usepackage{tikz}
\usetikzlibrary{positioning,backgrounds,fit,calc}

\usepackage[ruled,vlined,linesnumbered]{algorithm2e}

\usepackage[caption=false]{subfig}

\newtheorem{definition}{Definition}[section]
\newtheorem{lemma}{Lemma}

\DeclareMathOperator*{\argmax}{arg\,max}

\begin{document}

\preprint{APS/123-QED}

\title{Decoding Quantum LDPC Codes using Collaborative Check Node Removal}

\author{Mainak Bhattacharyya}
\email{mainak23@iiserb.ac.in}
\author{Ankur Raina}%
 \email{ankur@iiserb.ac.in}
\affiliation{%
 Department of Electrical Engineering and Computer Science, Indian Institute of Science Education and Research Bhopal, Bhopal 462066, India\\
}%




\date{\today}

\begin{abstract}
   Fault tolerance in quantum protocols requires contributions from error-correcting codes and their suitable decoders.
   Quantum Low-Density Parity Check (QLDPC) codes are one of the most explored quantum codes that have good coding rate and efficient decoders.
   Iterative message passing-based decoders, although fast, fail to produce suitable success rates due to the colossal degeneracy and short cycles intrinsic to these codes.
   In this work we present a strategy to improve the performance of the Belief Propagation (BP) decoding, specifically the min-sum algorithm.
   We propose a collaborative decoding framework that integrates message passing with stabilizer check node removals.
   We further introduce the concept of ``qubit separation" and show that the improved decoding performance is directly related to the generation of highly separated trapped data qubits.
   To guide a more selective removal of check nodes that constrain the separation of the trapped data qubits, we introduce information measurements (IMs) for the data qubits and their adjacent stabilizer checks.
   We evaluate the performance of the proposed collaborative decoder on Generalized Hypergraph Product (GHP) codes and demonstrate that appropriate decoder configurations mitigate trapping sets in min-sum decoding without significant overhead.
\end{abstract}

\maketitle


\section{\label{sec:introduction}Introduction}
Quantum Error correction (QEC) has been a primary focus of many to realize the quantum protocols with high fidelity.
Over the years, Quantum Low-Density Parity Check (QLDPC) codes have emerged as the most suitable candidate for achieving practical fault tolerance due to their low overhead \cite{gottesman2013fault} and efficient distance scaling \cite{panteleev2021quantum,tillich2013quantum}, establishing them as strong candidates for practical QEC.
In comparison with the classical LDPC codes, good QLDPC code constructions are less trivial.
The first set of construction towards achieving the same was given by Mackay $\emph{et al.}$ \cite{mackay2004sparse}.
Although one of the most famous constructions for good QLDPC codes is given by Tillich and Z{\'e}mor \cite{tillich2013quantum}.
These QLDPC codes called the Hypergraph Product (HGP) codes, achieve a quadratic distance scaling (that is, $d \propto \sqrt{n}$) as an outcome of the tensor product between any two classical codes and remained state-of-the-art QEC code for nearly a decade.
The well-known Toric code is a class of HGP code itself \cite{kitaev2003fault}.
More recent studies have presented the first constructions of asymptotically good QLDPC codes \cite{panteleev2022asymptotically,leverrier2022quantum}. 
For example, Panteleev and Kalachev proposed a generalization of the HGP construction \cite{panteleev2021degenerate}, which they later renamed Lifted Product codes \cite{panteleev2021quantum}, reporting the first QLDPC codes with an almost linear distance.\\\\ 
The performance of these QLDPC codes remains suboptimal under the currently available decoding algorithms.
Efficient decoders constitute a critical component of fault tolerant quantum computing architectures, as they are classical algorithms tasked with inferring suitable recovery operations from noisy syndrome measurements.
Unlike classical LDPC codes, the decoding of QLDPC codes becomes fundamentally difficult due to highly degenerate stabilizers and abundant short cycles present in the Tanner graphs of these codes.
Although very efficient, iterative message-passing algorithms like belief propagation (BP) decoders perform poorly on QLDPC codes.
A significant amount of research has been done to characterize these harmful configurations of quantum codes, namely the trapping sets (TS) \cite{raveendran2021trapping}.
Further, this knowledge of TS has been used extensively to design better message-passing-based decoders for QLDPC codes \cite{pradhan2023learning,liu2019neural,gong2024graph}.
Among the most effective solutions proposed to improve the performance of BP decoders, the BP+OSD decoder introduced in Ref. \cite{panteleev2021degenerate} remains state-of-the-art in terms of decoding accuracy. 
Although it aims to improve the outcome of BP by augmenting a classical post-processing step based on ordered statistics decoding (OSD), which solves a linear system to identify low-weight error configurations consistent with the measured syndrome. 
While highly accurate, even zeroth-order OSD post-processing incurs a computational complexity of $\Theta(n^3)$ in the number of physical qubits $n$. 
Higher-order OSD further improves decoding performance but rapidly becomes computationally prohibitive for any practical use. 
Recent works have sought to retain comparable decoding accuracy while reducing computational complexity \cite{hillmann2025localized,wolanski2024ambiguity}.\\\\
In this work, we propose a quantum collaborative check node removal (QCCNR) decoder.
We analyze the trapping set scenarios due to the min-sum BP algorithm.
Our decoding algorithm shows that carefully removing stabilizer check nodes from the computation tree of the BP decoder can help overcome the trapping sets of the quantum code.
Our contribution also includes a detailed analysis of how the carefully designed sub-routines, such as information measurement (IM), are applied to the stabilizer check node removals based on the degree of data qubits of the quantum code..
This achieves the feat of converging error beliefs for both the qubits belonging to classical and quantum-type trapping sets.
We justify the mitigation of quantum trapping sets by introducing a notion of qubit separation, which links the detrimental capacity of trapped qubits to the effect of removing certain harmful error beliefs from BP decoding iterations.
This removal of error beliefs is essentially the removal of stabilizer check nodes during message-passing operations, and we show that highly separated trapped qubits are a direct consequence of such stabilizer check node removal from the computation tree of the trapped data qubits.
One of the key advantages of our algorithm stems from the fact that its primary intent is to improve the message passing algorithm of the BP decoder itself and does not rely on expensive sub-routines like matrix-inversions of OSD \cite{panteleev2021degenerate} and SI \cite{du2022stabilizer}; based improvements to BP.\\\\
The paper is organized as follows: In Section \ref{sec:prelims}, we introduce the basics of classical, quantum error correction and also briefly cover the min-sum algorithm.
Section \ref{sec:trappingset-qubitseparation} includes the trapping set description of QLDPC codes.
Here we also introduce the notion of qubit separation and show how higher separation of the trapped qubits leads to improved convergence of the min-sum algorithm.
In section \ref{sec:collab-decoding}, we formulate the complete architecture of a collaborative decoding, which compiles the stabilizer check removal with the IM subroutine, while invoking the message passing algorithm.
Further, we show the effects of mitigating the quantum trapping set configurations through numerical memory experiments of certain QLDPC codes.

\section{Background}
\label{sec:prelims}
\subsection{Classical error correction}
Classical error correction typically represents encoding $k$ bits of information $p$ into $n$ bits of codeword $c$, where $n > k$.
The codespace $\mathcal{C}: c \in \mathcal{C}$ of the classical error correcting code represents the nullspace of an $m \times n$ matrix $\mathbf{H}$, called the parity check matrix, i.e.
\begin{align}
\label{eq:define-classical-code}
    \mathbf{H}\mathrm{c}^T = 0 \,\,\text{mod}\,\, 2,
\end{align}
where $m$ is the total number of parity check constraints.
Errors in classical bits typically flip the codeword bits and are detected via non-trivial syndromes.
For instance an error $\eta$, such that $\mathrm{c} \rightarrow \mathrm{y}:= \mathrm{c} + \eta$ violates Eq. \eqref{eq:define-classical-code}.
Therefore, the knowledge of the error $\eta$ is often deciphered using the syndrome $\mathrm{s}$, estimated using $\mathrm{y}$ and $\mathbf{H}$:
\begin{align}
\label{eq:syndrome-defn}
    \mathbf{H}\mathrm{y}^T = \mathrm{s}. 
\end{align}
The minimum distance $d$ of a classical code indicates that the classical code can correct any $\frac{d-1}{2}$ bit errors.
All these correctable errors have unique non-zero syndromes.
The classical codes are, in general, represented using the notation of $[n,k,d]$.
\\
The parity check matrices of any $[n,k,d]$ code are often visualized as a bipartite graph, called the Tanner graph.
It consists of two sets of nodes: one representing the codeword bits and the other representing the parity-check constraints.
Further, an edge exists between these two types of nodes if the corresponding codeword is associated with the corresponding parity check constraint.
These Tanner graphs play a crucial role in many decoding algorithms and in code design.
\subsection{Quantum error correction}
Quantum systems suffer from a continuous set of errors, unlike the classical case, where the only source of errors is bit flips at discrete positions.
These errors can be modeled as a random occurrence of a discrete set of Pauli operators $\{I, X, Y, Z\}$ \cite{knill1997theory}.
To protect a quantum state against Pauli errors, a quantum error correcting code (QEC) encodes a $k$ qubit quantum state $\ket{\psi}$ into an $n$ qubit logical state $\ket{\psi}_L$.
These logical states have an additional property that it remains invariant under the influence of a set of operators called the stabilizers, i.e., $r_j \ket{\psi}_L = (+1)\ket{\psi}_L, \forall r_j \in \mathcal{S}$.
Here $\mathcal{S}$ is the group of all the mutually commuting stabilizer operators.
Any error $E$ on the encoded logical state either commutes with the stabilizers or anticommutes with them.
In the latter case, the error can be detected and corrected, as it takes the encoded state out of the codespace.
For instance, assume the erroneous quantum state is $\ket{\psi}_E = E \ket{\psi}_L$ and the consequence of anticommutation of the error operator with any of the stabilizer maps the erroneous state outside the codespace, i.e. $r_j \ket{\psi}_E = -\ket{\psi}_E$.
Although the erroneous state remains an eigenstate of the stabilizer with eigenvalue $-1$.
This, along with the fact that the stabilizer operators are Hermitian, is exploited to detect errors via projective quantum measurements of the stabilizer operators.
These measurements do not disturb the erroneous state $\ket{\psi}_E$, as it is still an eigenstate of the stabilizers but reveals the eigenvalue indicating an error.
The collection of measurement outcomes of all the stabilizer operators $\mathrm{s} = \{s_1, s_2, ..., s_{n-k}\}$ is called the syndrome.
\\
There are errors $E$, which commute with all the stabilizer operators.
This poses a dilemma as the erroneous state in this case still belongs to the codespace, i.e., $r_j \ket{\psi}_E = \ket{\psi}_E$.
Therefore, error $E$ cannot be detected via measurement of the stabilizers.
These errors belong to the centralizer group defined as $\mathcal{C}(\mathcal{S}) = \{E \in \mathcal{P}_n : Er_j = r_jE, \forall r_j \in \mathcal{S}\}$.
Although if $E \in \mathcal{C}(\mathcal{S}) \backslash \mathcal{S}$, then it acts non-trivially in the codespace and is still undetectable.
These errors have detrimental effects and are called logical errors.
Note that the codespace still remains invariant under the logical errors.
Minimum distance of these quantum codes is defined as the minimum weight of a logical error, i.e., $d = \text{min}\{|E|: E \in \mathcal{C}(\mathcal{S}) \backslash \mathcal{S}\}$.
All the QEC codes in general are represented using the notation $[[n,k,d]]$.

\subsection{Good QLDPC codes from ansatz}
Similar to classical codes, a parity check matrix (PCM) can be obtained for a $[[n,k,d]]$ quantum code using a binary mapping of the $n-k$ stabilizer operators.
It maps each of the $n$ qubit Pauli operators into a $2n$ bit string, as follows:
\begin{align}
    \bigotimes_{i=1}^N X^{x_i}Z^{z_i} \mapsto (x_1, ..., x_n | z_1, ..., z_n) : x_i,z_i \in \{0,1\}.
\end{align}
The above map is used, and the resultant PCM of size $(n-k \times 2n)$ has the form:
\begin{align}
    \label{eq:symplectic-quantum-pcm}
    \mathbf{H} = \begin{bmatrix}
        \mathbf{H}_X | \mathbf{H}_Z
    \end{bmatrix},
\end{align}
where each row of $\mathbf{H}$ is a binary representation of the stabilizers.
Therefore, from the mutual commutative nature of the stabilizers, it is easy to conclude the following condition:
\begin{align}
\label{eq:stabilizer-commute-formula}
    \mathbf{H}_X\mathbf{H}^{T}_Z + \mathbf{H}_Z\mathbf{H}^{T}_X = \mathbf{0}.
\end{align}
A particularly popular class of quantum codes namely the Calderbank-Shor-Steane (CSS) code \cite{calderbank1996good, steane1996error} has a special property that the non-identity operators of the stabilizer operators are of either Pauli-X or Pauli-Z type.
This effectively reduces Eq. \ref{eq:symplectic-quantum-pcm} into the following form:
\begin{align}
    \mathbf{H} = \begin{bmatrix}
        \mathbf{H}_X & \mathbf{0}\\
        \mathbf{0} & \mathbf{H}_Z
    \end{bmatrix},
\end{align}
Further it follows that Eq. \ref{eq:stabilizer-commute-formula} reduces down to
\begin{align}
    \label{eq:css-commute-formula}
    \mathbf{H}_X\mathbf{H}^{T}_Z = \mathbf{0}.
\end{align}
The Quantum Low-Density Parity Check (QLDPC) code family is a special set of CSS codes with the notable property of a sparse parity-check matrix.
This implies that the row and column weights of the PCM are upperbounded by some constants.
In this work, we particularly deal with two recently emerged families of QLDPC codes, namely the Generalized Hypergraph Product (GHP) codes \cite{panteleev2021degenerate} and the Generalized Bicycle (GB) codes \cite{kovalev2013quantum}.
Both of these codes originate through a clever exploitation of the commutativity conditions in Eq. \ref{eq:css-commute-formula}.
Panteleev and Kalachev discuss the use of large families of matrices called ansatz, which already satisfy Eq. \ref{eq:css-commute-formula}.
For instance consider to ansatz matrices $A$ and $B$, such that $AB=BA$ and 
\begin{align}
    \label{eq:pcm-gb-code}
    H_X = [A, B] , H_Z = [B^T, A^T].
\end{align}
From the above equation and the commutativity of the ansatzes, we observe $H_XH_Z^{T} = AB + BA = \mathbf{0}$.
Kovalev and Pryadko in \cite{kovalev2013quantum} proposed the use of binary circulant matrices for the construction of GB codes.
Onae similar note Panteleev and Kalachev proposed the GHP codes in \cite{panteleev2021degenerate}, as follows:
\begin{align}
\label{eq:pcm-ghp-code}
    \mathbf{H}_X = [A,bI_m] \,\,\text{and}\,\, \mathbf{H}_Z = [b^TI_n,A^T].
\end{align}
The primary element in this construction is the elements from a ring $R$ of $l \times l$ circulant matrices over the field $\mathbb{F}_2$ \cite{panteleev2021degenerate}.
The ring of all $m \times n$ matrices over ring $R$ is denoted by $\mathcal{M}_{m \times n}(R)$ or,  $\mathcal{M}_{n}(R)$, if $m = n$.
In Eq. \ref{eq:pcm-ghp-code}, $b \in \mathcal{M}_l(\mathbb{F}_2)$ and $A$ is a block matrix, such that $A \in \mathcal{M}_{m \times n}(R)$ or, $A \in \mathcal{M}_{ml \times nl}(\mathbb{F}_2)$.
In Appendix \ref{ap:codes} we describe in detail the GHP and GB codes used in this work to analyze the trapping set mitigation and decoder performance benchmarking.

\subsection{Belief propagation and min-sum algorithm}
The decoding problem of a quantum code can be described in terms of finding a correction operator $\hat{E}$ from the most-probable coset belonging to the coset partition of the $n$-fold Pauli group $\mathcal{G}_n$, which satisfies the measured syndrome $\mathrm{s}$.
This decoding problem for QLDPC codes is fundamentally different from its classical counterpart, as the latter deals with finding only the most-probable error vector, i.e. 
\begin{align}
    \hat{E} = \argmax_{E \in \mathcal{G}_n}P(E|\mathrm{s}).
\end{align}
For a uniformly distributed random noise model, the above problem can be further defined as determining the marginal probabilities of errors at each individual bit, i.e.,
\begin{align}
    P(e_i) = \sum_{\sim e_i} P(e_1, e_2, ..., e_i = 1, ..., e_n | \mathrm{s}),
\end{align}
where, $\sim e_i$ denotes a sum over all the $e_j$ except $e_i$.
Belief Propagation (BP) is an iterative message-passing algorithm that computes exact marginal distributions on tree-structured factor graphs and approximate marginals on loopy graphs.
The BP decoder takes the parity-check matrix and syndrome information as input and iteratively updates the error beliefs for each possible error location.
\\
In this work, we use the min-sum algorithm of BP for QLDPC codes.
The CSS nature of these codes ensures that the decoding problem for Pauli-X and Pauli-Z type errors can be addressed independently.
We use $v_i$ and $c_j$ to denote the $i^{th}$ data qubit node and $j^{th}$ stabilizer check nodes of the Tanner graph, respectively.
Also, $m^{k}_{a \rightarrow b}$ is the message passed at the $k^{th}$ iteration of the algorithm from node $a$ to node $b$ on the Tanner graph.
Following we describe an overall brief summary of the standard message passing rule of the min-sum BP algorithm:
\begin{enumerate}
    \item Initialization: Initially from each data qubits message $m^0_{v_i\rightarrow c_j}$ is passed as follows:
    \begin{align}
        m^0_{v_i\rightarrow c_j} = \lambda_i = \log(\frac{1-p_i}{p_i});\,\, \forall i,j : \mathbf{H}_{ji} = 1.
    \end{align}
    \item Check to data qubit messages: For the subsequent iterations, message $m^{k}_{c_j \rightarrow v_i}$ is passed from each check node $c_j$ to data qubits $v_i$.
    In such messages, we exclude the contribution of the message passed from the data qubit node $v_i$ to check qubit node $c_j$ in the previous ${k-1}^{th}$ iteration.
    The check $c_j$'s belief that error at $v_i$ has occurred is expressed as:
    \begin{align}
        m^{k}_{c_j \rightarrow v_i} = (-1)^{\mathrm{s}_j}\alpha \left(\prod_{v_p:\,N(c_j)\backslash v_i} \mathrm{sign}(m^{k-1}_{v_p\rightarrow c_j})\right) w,
    \end{align}
    where $\alpha$ is a scaling factor and can affect the convergence.
    $\mathrm{s}_j$ is the syndrome value corresponding to the check node $c_j$ and $w = \min_{v_p:\,N(c_j)/v_i}\{|m^{k-1}_{v_p\rightarrow c_j}|\}$
    Also, here and in the rest of the article, we denote the set of neighborhoods of any data qubit $q_i$ or stabilizer check $c_j$ as $N(q_i)$, $N(c_j)$, respectively.
    \item Data qubit to check messages: Except for the initialization as described, at each iteration $k$, the messages passed from data qubits to check nodes follows:
    \begin{align}
        m^k_{v_i\rightarrow c_j} = \lambda_i + \sum_{c_p: N(v_i)\backslash c_j}m^{k-1}_{c_p\rightarrow v_i}
    \end{align}
    \item Hard decision: At the end of each iteration, the marginal is computed at each data qubit node $v_i$ as:
    \begin{align}
        \gamma_i \gets \lambda_i + \sum_{c_p: N(v_i)}m^{k}_{c_p\rightarrow v_i}.
    \end{align}
    Further, a hard decision is made from the newly computed marginals.
    If $\mathrm{sign}(\gamma_i) = -1$, then the hard decision for data qubit $i$ is $\hat{e}_i = 1$, otherwise $\hat{e}_i = 0$.
\end{enumerate}
According to the hard decisions taken at the end of each iteration, if $\hat{e} = \{e_i\}$ satisfies the syndrome, then the BP decoder has converged and $\hat{e}$ is returned; otherwise, the decoder moves on to the next message passing iteration, until a predetermined number of iterations have been completed.
\\
Although the above algorithm can not be used to decode quantum codes without any modifications.
For instance, the BP decoder cannot converge if there exist degenerate error configurations like $E_1$ and $E_2$, such that $E_1 + E_2 \in \text{rowspace}(\mathbf{H})$.
The phenomena leading to such scenarios are known as split beliefs.
In the later sections, we formally characterize such harmful configurations of quantum codes and propose the metric of qubit separation to discuss more intuitively an approach of removing redundant check nodes from the message passing sub-routines of BP to assist the convergence of the decoder.
\section{Mitigating the trapping sets of QLDPC codes using qubit separation}
\label{sec:trappingset-qubitseparation}
Iterative message passing decoders like min-sum BP are vulnerable when it comes to certain error configurations.
Certain topologies of sub-graphs containing $a$ number of data qubits and $b$ number of odd degree check nodes contributes to these detrimental behavior of the min-sum decoder, is called a $(a,b)$ trapping set (TS).
In this section we describe the origin of these trapping sets for the $[[882,24]]$ GHP code and use a novel metric of qubit separation to argue that stabilizer check nodes if removed from the computation tree of the message passing algorithm of BP, can contribute to mitigation of trapping sets.
Construction details of the $[[882,24]]$ GHP code can be found in Appendix \ref{ap:codes}.
\subsection{Trapping sets of $[[882,24]]$ GHP code}
For QLDPC codes, there exist two classes of harmful TS configurations: Classical-type trapping sets (CTS) and Quantum trapping sets (QTS) \cite{raveendran2021trapping}.
\begin{definition}[Classical-type Trapping sets]
    \label{def:classical trapping set}
    A classical-type trapping set is a set of qubits that either do not converge or
    are adjacent to unsatisfied stabilizer checks after a predetermined number of iterations of a syndrome-based iterative decoder.
\end{definition}
For example in Fig. \ref{fig:3_3-cts}, we show a $(3,3)$ CTS obtained from $\mathbf{H}_Z$ of the $[[882,24]]$ GHP code.
\begin{figure}[t]
\centering
\subfloat[Classical trapping set\label{fig:3_3-cts}]{
\centering
\resizebox{0.4\columnwidth}{!}{%
\begin{tikzpicture}[
    node distance=0.95cm,
    scale=0.95,
    every node/.style={font=\footnotesize},
    var/.style={circle, draw=black, fill=blue!25, minimum size=6mm, inner sep=1.2pt},
    unsat/.style={rectangle, draw=black, fill=red!25, minimum size=5.8mm, inner sep=1.2pt},
    sat/.style={rectangle, draw=black, fill=black!15, minimum size=5.6mm, inner sep=1.1pt},
    edge/.style={draw=black, line width=0.7pt}
]

\node[var] (v0) at (0, 0) {$v_0$};
\node[unsat] (c0) [above=of v0] {$c_0$};
\node[sat] (c1) [right=of v0] {$c_1$};
\node[var] (v1) [right=of c1] {$v_1$};
\node[unsat] (c2) [above=of v1] {$c_2$};
\node[sat] (c7) [below=of v1] {$c_7$};
\node[sat] (c6) [below=of v0] {$c_6$};
\node[var] (v6) [below=of c1] {$v_6$};
\node[unsat] (c12) [below=of v6] {$c_{12}$};

\draw[edge] (v0) -- (c0);
\draw[edge] (v0) -- (c1);
\draw[edge] (v0) -- (c6);
\draw[edge] (v1) -- (c1);
\draw[edge] (v1) -- (c2);
\draw[edge] (v1) -- (c7);
\draw[edge] (v6) -- (c6);
\draw[edge] (v6) -- (c7);
\draw[edge] (v6) -- (c12);

\end{tikzpicture}
}
}
\vfill
\subfloat[Quantum trapping set\label{fig:3-3-qts}]{
\centering
\resizebox{0.85\columnwidth}{!}{%
\begin{tikzpicture}[
    node distance=0.95cm,
    scale=0.95,
    every node/.style={font=\footnotesize},
    varA/.style={circle, draw=black, fill=blue!25, minimum size=6mm, inner sep=1.2pt},
    varB/.style={circle, draw=black, fill=violet!25, minimum size=6mm, inner sep=1.2pt},
    check/.style={rectangle, draw=black, fill=black!15, minimum size=5.8mm, inner sep=1.2pt},
    edge/.style={draw=black, line width=0.7pt},
    dottededge/.style={draw=black, dotted, line width=0.65pt}
]

\node[varA] (v0) at (0, 0) {$v_0$};
\node[check] (c0) [right=of v0] {$c_0$};
\node[varB] (v477) [right=of c0] {$v_{477}$};
\node[check] (c405) [below right=of v477] {$c_{405}$};
\node[varA] (v405) [below right=of c405] {$v_{405}$};
\node[check] (c406) [below left=of v405] {$c_{406}$};
\node[varB] (v478) [below left=of c406] {$v_{478}$};
\node[check] (c352) [left=of v478] {$c_{352}$};
\node[varA] (v351) [left=of c352] {$v_{351}$};
\node[check] (c357) [above left=of v351] {$c_{357}$};
\node[varB] (v483) [above left=of c357] {$v_{483}$};
\node[check] (c6) [above right=of v483] {$c_{6}$};

\node[check] (c1) at ($(v0)!0.75!(v478)$) {$c_{1}$};
\node[check] (c411) at ($(v483)!0.25!(v405)$) {$c_{411}$};
\node[check] (c351) at ($(v351)!0.75!(v477)$) {$c_{351}$};

\draw[edge] (v0) -- (c0);
\draw[edge] (c0) -- (v477);
\draw[edge] (v477) -- (c405);
\draw[edge] (c405) -- (v405);
\draw[edge] (v405) -- (c406);
\draw[edge] (c406) -- (v478);
\draw[edge] (v478) -- (c352);
\draw[edge] (c352) -- (v351);
\draw[edge] (v351) -- (c357);
\draw[edge] (c357) -- (v483);
\draw[edge] (v483) -- (c6);
\draw[edge] (c6) -- (v0);

\draw[dottededge] (v478) -- (c1);
\draw[dottededge] (v0) -- (c1);
\draw[dottededge] (v483) -- (c411);
\draw[dottededge] (v405) -- (c411);
\draw[dottededge] (v351) -- (c351);
\draw[dottededge] (c351) -- (v477);

\end{tikzpicture}
}
}
\caption{In (a) we show a typical $(3,3)$ 6-cycle classical type trapping set \cite{chytas2024collective} and in (b) is a $(6,0)$ Quantum Trapping set of the $[[882,24]]$ GHP code \cite{chytas2024collective}.}
\label{fig:trapping-sets-of-882-24}
\end{figure}
The qubit and check node indexing shows that the $(3,3)$ CTS is formed due to the contributions of the polynomial $b$ and can be tracked down from the matrix $b^T I_n$ \cite{chytas2024collective}.
The \textit{min-sum} BP decoder fails to converge in this scenario due to the symmetric message passing rules.
For instance, a syndrome with support on $\mathrm{supp}(\mathrm{s}) = \{c_0,c_1,c_2,c_6,c_7,c_{12}\}$, indicating the violated stabilizer checks, will lead to an oscillation for the \textit{min-sum} BP.
The BP decoder output oscillates between the data qubit nodes $\{v_0, v_1, v_6\}$ and an all zero error pattern.
\\\\
The other type of TS, specific to the QLDPC codes is the Quantum Trapping sets (QTS).
\begin{definition}[Quantum Trapping sets]
    \label{def:quantum trapping sets}
    A quantum trapping set is a collection of an even number of data qubits or, symmetric stabilizers, such that the induced sub-graph contains no odd-degree stabilizer check nodes. 
    This set of nodes form two degenerate subsets of error patterns of equal weight with a set of common odd-degree stabilizer (check) neighborhood.
\end{definition}
A QTS with $a$ number of data qubits is referred as a $(a,0)$ trapping set.
Fig. \ref{fig:3-3-qts} shows a typical QTS formed for the $[[882,24]]$ GHP code with $6$ data qubits in the QTS \cite{chytas2024collective}.
The two degenerate error patterns of this $(6,0)$ QTS can be identified as $V_a = \{v_0, v_{351}, v_{405}\}$ and $V_b = \{v_{477}, v_{478}, v_{483}\}$.
Further, as per the definition \ref{def:quantum trapping sets}. $V_a$ and $V_b$ have the same set of odd-degree stabilizer check neighborhood. That is, $N(V_a) = N(V_b) = \{c_0, c_1, c_6, c_{405}, c_{406}, c_{411}, c_{351}, c_{352}, c_{367}\}$.
\begin{lemma}
    \label{lem:qts}
    Iterative decoders with critical number $\frac{a}{2}$ posses no TS characteristics for the $(a,0)$ QTS, if the cardinality of the error subsets $e_a \subseteq V_a$ or $e_b \subseteq V_b$ exceeds $\frac{a}{2}$ \cite{raveendran2021trapping}.
\end{lemma}
The critical number therefore denotes the largest possible error cardinality for which message-passing decoders, such as min-sum BP fail when the error pattern is supported on the QTS. 
Following the Lemma \ref{lem:qts}, error patterns of cardinality less or equal $\frac{a}{2} = 3$, on any of the degenerate set is a harmful configuration for the min-sum BP decoder.
For instance any $e_a \subseteq V_a : |e_a| \leq |V_a|$ has an exact twin $e_b \subseteq V_b$.
This results into an oscillation between $e_a \oplus e_b$ and an all zero error at each iterations of the \textit{min-sum} decoder.
\subsection{Convergence of min-sum BP for highly separated trapped qubits}
The oscillatory behavior of the BP decoder is majorly contributed by the error beliefs of certain stabilizer check nodes.
For instance, the error beliefs from the stabilizer check nodes with even number of neighboring data qubits in both the CTS and QTS is detrimental to BP convergence.
We also observe that the check nodes contributing to the harmful error beliefs are also responsible for limiting the separation of trapped qubits.
In our example of a $(3,3)$ CTS and $(6,0)$ QTS, these check nodes are $\{c_1, c_6, c_7\}$ and $\{c_0, c_1, c_6, c_{405}, c_{406}, c_{411}, c_{351}, c_{352}, c_{367}\}$ respectively.
Therefore, a natural intuition is to remove the message contributions of these check nodes from the marginal estimations of the corresponding data qubits to ensure better convergence of the BP decoder.\\
The message-passing dynamics of the BP decoder can be more systematically understood through the framework of computation trees \cite{frey2001signal,wiberg1996codes}.
In the context of quantum error correction, we adopt the following definition of computation tree (CT):
\begin{definition}[Computation Tree]
    \label{def:cmputation-tree}
    The computation tree $\mathcal{T}_\mathcal{K}(t)$ for an iterative decoder on a code is constructed by picking a root node `$t$' corresponding to either a qubit or a stabilizer check node and then iteratively adding `$\mathcal{K}$' levels of edges and leaf nodes, where each level represents a complete iterations (i.e., qubit (check) $\rightarrow$ check (qubit) $\rightarrow$ qubit (check)) of the messages passed for the iterative decoder.
\end{definition}

\begin{figure}[t]
\centering
\begin{tikzpicture}[
    node distance=0.95cm,
    scale=0.82,
    transform shape,
    every node/.style={font=\footnotesize},
    varA/.style={circle, draw=black, fill=blue!25, minimum size=6mm, inner sep=1.2pt},
    varB/.style={circle, draw=black, fill=green!25, minimum size=6mm, inner sep=1.2pt},
    unsat/.style={rectangle, draw=black, fill=red!25, minimum size=5.8mm, inner sep=1.2pt},
    sat/.style={rectangle, draw=black, fill=black!15, minimum size=5.8mm, inner sep=1.2pt},
    edge/.style={draw=black, line width=0.7pt},
    rededge/.style={draw=red!70!black, line width=0.7pt},
    dottededge/.style={draw=black, dotted, line width=0.65pt}
]


\node[varA] (v0) at (0, 0) {$v_0$};
\node[unsat] (c0) [below=of v0] {$c_0$};
\node[sat] (c1) [left=of c0] {$c_1$};
\node[sat] (c6) [right=of c0] {$c_6$};

\node[varB] (v477) [below=of c0] {$v_{477}$};
\node[varB] (v62) [left=of v477] {$v_{62}$};
\node[varB] (v57) [left=of v62] {$v_{57}$};
\node[varB] (v513) [right=of v477] {$v_{513}$};
\node[varB] (v567) [right=of v513] {$v_{567}$};

\node[sat] (c5) [below left=of v62] {$c_5$};
\node[sat] (c62) [below right=of v62] {$c_{62}$};

\node[varB] (v5) [below right=of c5] {$v_{5}$};
\node[varB] (v4) [below left=of c5] {$v_{4}$};

\node[sat] (c6_a) [below right=of v5] {$c_{6}$};
\node[sat] (c11) [below left=of v5] {$c_{11}$};

\node[varA] (v6) [below right=of c6_a] {$v_6$};

\node[varB] (v10) [below left=of c11] {$v_{10}$};
\node[varB] (v11) [right=of v10] {$v_{11}$};

\node[varB] (v56) [below right=of c62] {$v_{56}$};    
\node[sat] (c57) [below right=of v56] {$c_{57}$}; 
\node[varB] (v57_a) [below right=of c57] {$v_{57}$};  


\draw[edge] (c0) -- (v0);
\draw[edge] (c1) -- (v0);
\draw[edge] (c6) -- (v0);
\draw[edge] (c0) -- (v57);
\draw[rededge] (c0) -- (v62);
\draw[edge] (c0) -- (v477);
\draw[edge] (c0) -- (v513);
\draw[edge] (c0) -- (v567);
\draw[rededge] (c5) -- (v62);
\draw[edge] (c62) -- (v62);
\draw[rededge] (c5) -- (v5);
\draw[edge] (c5) -- (v4);
\draw[edge] (c11) -- (v5);
\draw[rededge] (c6_a) -- (v5);
\draw[rededge] (c6_a) -- (v6);
\draw[edge] (c11) -- (v10);
\draw[edge] (c11) -- (v11);
\draw[edge] (c62) -- (v56);
\draw[edge] (c57) -- (v56);
\draw[edge] (c57) -- (v57_a);


\node (v1r) [right=of v567] {};
\node (v1l) [left=of v57] {};
\node (v2r) [right=of v5] {};
\node (v3l) [below left=of c6_a] {};

\draw[dottededge] (c6) -- (v1r);
\draw[dottededge] (c1) -- (v1l);
\draw[dottededge] (c5) -- (v2r);
\draw[dottededge] (c6_a) -- (v3l);


\draw[dottededge, line width=1pt] (-5, -3.7) -- (5, -3.7);
\draw[dottededge, line width=1pt] (-5, -6.1) -- (5, -6.1);
\draw[dottededge, line width=1pt] (-5, -8.5) -- (5, -8.5);

\node at (4.5,-3.4) {\textbf{\textcolor{black}{1}}};
\node at (4.5,-5.8) {\textbf{\textcolor{black}{2}}};
\node at (4.5,-8.2) {\textbf{\textcolor{black}{3}}};

\end{tikzpicture}

\caption{Computation tree of the trapped qubit $v_0$ of CTS $(3,3)$.
Each level of the computation tree are indicated through the dotted line, representing each complete iteration of the BP decoder.
In level $3$ we observe the data qubit node $v_6$, which belongs to the same $(3,3)$ CTS.
This limits the qubit separation of $v_0$ to $q_k = 2$.
}
\label{fig:computation-tree-3-3}
\end{figure}
In Fig. \ref{fig:computation-tree-3-3} we show a level $3$ computation tree $\mathcal{T}_3(v_0)$ constructed from the message passing pattern of the min-sum BP on the $[[882,24]]$ GHP code.
Each level of $\mathcal{T}_3(v_0)$ represents a  complete iteration of the decoder.
\\
Analyzing the BP decoder's performance using CT helps us introduce the notion of qubit separation.
We argued before that exclusion of the error beliefs from harmful stabilizer check nodes can improve the BP decoder's performance.
This essentially implies removal of certain check nodes from the computation tree of the decoder.
We now claim that this procedure is not arbitrary and is directly linked to a phenomena, which generates highly separated trapped data qubits.
In classical coding theory, a similar approach is used to improve the performance of the min-sum algorithm to break the trapping sets \cite{kang2015breaking}.
We extend those concepts to the non-trivial case of quantum trapping sets.
Firstly we identify that the separation of qubits should be defined separately for the two type of trapping sets in the quantum codes.
Suppose, $\mathcal{T}_s$ denotes the set of data qubits that are trapped in any TS.
The separation of a qubit trapped in a classical-type trapping set is defined as follows:
\begin{definition}[Qubit separation for CTS]
    \label{def:classical-separation}
    If an erroneous data qubit node $v \in V_1$ for any $V_1 \subset \mathcal{T}_s$; has at least one degree-one check node $c$, such that, within $\mathcal{K}$ number of message passing iterations of the decoder, there exists no more data qubit $u \in V_1$ as a descendant of $c$ in $T_{\mathcal{K}}(c)$, then the root data qubit $v$ is said to be $\mathcal{K}$ separated.
\end{definition}
We use $q_k$ to denote any qubit separation.
In Fig. \ref{fig:computation-tree-3-3} we can observe that the trapped qubit $v_0$ belonging to the $(3,3)$ CTS has a separation $q_k = 2$.
\\
Now the immediate predecessor stabilizer check node connected to $v_6$ at level $\mathcal{K}=3$ is $c_6$.
Therefore, removal of $c_6$ from level $\mathcal{K} = 3$, effectively increases the qubit separation of $v_0$.
In Fig. \ref{fig:computation-tree-3-3}, we have shown one of such paths that leads to the limiting separation of the trapped qubit.
Although there are exactly two paths, which leads to $q_k = 2$ for $v_0$.
Therefore, removal of all those corresponding precursor check nodes in level $\mathcal{K} = 3$ of the computation tree, increases the qubit separation of $v_0$.
\\\\
The origin of the QTS is different than the CTS.
A QTS has no odd-degree stabilizer check nodes.
Therefore, the qubit separations for data qubits in a QTS can not be defined similar to the case of a CTS.
Consider Lemma \ref{lem:qts} and the degenerate subset of data qubits in the QTS.
The common stabilizer neighborhood of the subsets essentially gets violated for error patterns, which generates the QTS characteristics.
The definition of the qubit separation for QTS, therefore  follows:
\begin{definition}[Qubit separation for QTS]
    \label{def:quantum-separation}
    If a QTS has degenerate subsets $V_a$ and $V_b$, then
    for an erroneous data qubit node $v \in V_a$ with at least one degree-one (w.r.t. the same subset data qubit nodes) check node $c$, is said to be $\mathcal{K}$ separated, if within $\mathcal{K}$ number of message passing iterations of the decoder, there exists no more data qubit $u \in V_a$ as a descendant of $c$ in $T_{\mathcal{K}}(c)$.
\end{definition}
The same definition extends towards the separation of qubits in any of the degenerate subset of qubits from the QTS.
For instance, consider the two degenerate subsets $V_a$ and $V_b$ of the $(6,0)$ QTS discussed previously.
The common stabilizer check neighborhood $N(V_a)$ or, $N(V_b)$ of these two subsets is responsible for both limiting the separation of the trapped data qubits and contributing to the violated stabilizer nodes when the error is supported on a harmful error pattern.

Therefore, in the computation tree of a data qubit trapped inside a QTS, the removal of check nodes should happen from the level $1$ itself, to increase the root data qubit's separation.
We now state 
\begin{lemma}
    A QLDPC code with data qubits having degree $d_v$, has a total of $d_v(d_v - 1)$ number of stabilizer check nodes $c_r: c_r \in N(V_a) = N(V_b)$, which limits the separation of any erroneous data qubit $v_r \in V_a$ or, $v_r \in V_b$ with $V_a$ and $V_b$ together forming the QTS.
    \label{cor:cor-no-check-rem}
\end{lemma}
\begin{proof}
    The proof of the above is very simple and is a direct consequence of the isomorphic property of the QTS.
    In a typical QTS, there are no odd-degree stabilizer check nodes.
    All the data qubits $v_r \in V_i$ belonging to the QTS exhibit $d_v$ number of adjacent stabilizer check nodes in layer $1$ of $T(v_r)$, and each of these check nodes has exactly one leaf data qubit $v_s \in V_j : i \neq j$ at level $1$ of $T(v_r)$.
    Further from the definition of QTS, we can conjecture that the isomorphic nature of the two disjoint subsets results in a set of data qubits at level $2$ of $T(v_r) \forall v_r \in V_i$, which limits the qubit separation.
    We assume these qubits are child nodes of the parent $p$ number of check nodes.
    Each leaf data qubit nodes at level $1$ of $T(v_r)$ has $d_v - 1$ number of child check nodes at level $2$, and some of these child check nodes at level $2$, has a descendant leaf data qubits $v_{r^\prime}: r^\prime \in V_i, r \neq r^{\prime}$, which contributes to limiting separation for trapped data qubit $v_r$.
    So, $p = q(d_v - 1)$, where $q$ is the number of eligible leaf data qubit nodes at level $1$ of the CT, such that they have $v_{r^\prime} \in V_i$ as descendants at level $2$ of $T(v_r)$.\\
    This number is simply $q = d_v$, because the root data qubit has $d_v$ number of child check nodes, and from each such check node, one of the descendant qubits in level $1$ belongs to the other isomorphic qubit subset of the QTS compared to the root data qubit node.
    Therefore, in total $d_v(d_v - 1)$ number of check nodes at level $2$ are responsible for the limiting separation of root trapped data qubit $v_r$.
\end{proof}
\begin{figure}[t]
\centering

\tikzset{
blue node/.style={circle, draw=black, fill=blue!25, minimum size=6mm, inner sep=1.2pt},
violet node/.style={circle, draw=black, fill=violet!25, minimum size=6mm, inner sep=1.2pt},
red node/.style={rectangle, draw=black, fill=red!25, minimum size=5.8mm, inner sep=1.2pt},
black node/.style={rectangle, draw=black, fill=black!15, minimum size=5.8mm, inner sep=1.2pt},
line/.style={draw=black, line width=0.7pt},
dot_line/.style={draw=black, dotted, line width=0.65pt}
}


\subfloat[]{
\resizebox{0.4\textwidth}{!}{
\begin{tikzpicture}[node distance=1cm, every node/.style={font=\footnotesize}]

\node[blue node] (v0) at (0,0) {$v_0$};
\node[red node] (c0) [below=of v0] {$c_0$};
\node[red node] (c1) [left=of c0] {$c_1$};
\node[red node] (c6) [right=of c0] {$c_6$};

\node[violet node] (v478) at (-3.5,-3.5) {$v_{478}$};
\node[violet node] (v477) [below=of c0] {$v_{477}$};
\node[violet node] (v483) at (3.5,-3.5) {$v_{483}$};

\path[line] (v0)--(c0);
\path[line] (v0)--(c1);
\path[line] (v0)--(c6);
\path[line] (v477)--(c0);
\path[line] (v478)--(c1);
\path[line] (v483)--(c6);

\node[black node] (c406) [below left=of v478] {$c_{406}$};
\node[black node] (c352) [below right=of v478] {$c_{352}$};
\node[blue node] (v405) [below=of c406] {$v_{405}$};
\node[blue node] (v351) [below=of c352] {$v_{351}$};

\node[black node] (c351) [below left=of v477] {$c_{351}$};
\node[black node] (c405) [below right=of v477] {$c_{405}$};
\node[blue node] (v405b) [below=of c405] {$v_{405}$};
\node[blue node] (v351b) [below=of c351] {$v_{351}$};

\node[black node] (c411) [below left=of v483] {$c_{411}$};
\node[black node] (c357) [below right=of v483] {$c_{357}$};
\node[blue node] (v405a) [below=of c411] {$v_{405}$};
\node[blue node] (v351a) [below=of c357] {$v_{351}$};

\path[line] (c406)--(v478);
\path[line] (c352)--(v478);
\path[line] (c406)--(v405);
\path[line] (c352)--(v351);
\path[line] (v483)--(c411);
\path[line] (c357)--(v483);
\path[line] (c411)--(v405a);
\path[line] (c357)--(v351a);
\path[line] (v477)--(c351);
\path[line] (v477)--(c405);
\path[line] (c405)--(v405b);
\path[line] (c351)--(v351b);

\path[dot_line] (c1)--++(-1.5,-1);
\path[dot_line] (c1)--++(0.5,-1);
\path[dot_line] (c0)--++(-0.5,-1);
\path[dot_line] (c0)--++(0.5,-1);
\path[dot_line] (c6)--++(1.5,-1);
\path[dot_line] (c6)--++(-0.5,-1);

\draw[dot_line, line width=1pt] (-5,-4.2)--(6,-4.2);
\draw[dot_line, line width=1pt] (-5,-7.4)--(6,-7.4);

\node at (5.4,-3.7) {\textbf{1}};
\node at (5.4,-6.8) {\textbf{2}};

\end{tikzpicture}}}
\vfill


\subfloat[]{
\resizebox{0.4\textwidth}{!}{
\begin{tikzpicture}[node distance=1cm, every node/.style={font=\footnotesize}]

\node[blue node] (v405) at (0,0) {$v_{405}$};
\node[red node] (c405) [below=of v405] {$c_{405}$};
\node[red node] (c406) [left=of c405] {$c_{406}$};
\node[red node] (c411) [right=of c405] {$c_{411}$};

\node[violet node] (v478) at (-3.5,-3.5) {$v_{478}$};
\node[violet node] (v477) [below=of c405] {$v_{477}$};
\node[violet node] (v483) at (3.5,-3.5) {$v_{483}$};

\path[line] (v405)--(c405);
\path[line] (v405)--(c406);
\path[line] (v405)--(c411);
\path[line] (v477)--(c405);
\path[line] (v478)--(c406);
\path[line] (v483)--(c411);

\node[black node] (c352) [below right=of v478] {$c_{352}$};
\node[black node] (c1) [below left=of v478] {$c_{1}$};
\node[black node] (c357) [below right=of v483] {$c_{357}$};
\node[black node] (c6) [below left=of v483] {$c_{6}$};

\node[blue node] (v0a) [below=of c1] {$v_{0}$};
\node[blue node] (v351) [below=of c352] {$v_{351}$};
\node[blue node] (v351a) [below=of c357] {$v_{351}$};
\node[blue node] (v0b) [below=of c6] {$v_{0}$};

\node[black node] (c0) [below left=of v477] {$c_{0}$};
\node[black node] (c351) [below right=of v477] {$c_{351}$};

\node[blue node] (v0c) [below=of c0] {$v_{0}$};
\node[blue node] (v351c) [below=of c351] {$v_{351}$};

\path[line] (c1)--(v478);
\path[line] (c352)--(v478);
\path[line] (v351)--(c352);
\path[line] (v0a)--(c1);

\path[line] (v483)--(c6);
\path[line] (v483)--(c357);
\path[line] (v0b)--(c6);
\path[line] (v351a)--(c357);

\path[line] (c0)--(v477);
\path[line] (c351)--(v477);
\path[line] (c0)--(v0c);
\path[line] (c351)--(v351c);

\path[dot_line] (c406)--++(-1.5,-1);
\path[dot_line] (c406)--++(0.5,-1);
\path[dot_line] (c405)--++(-0.5,-1);
\path[dot_line] (c405)--++(0.5,-1);
\path[dot_line] (c411)--++(1.5,-1);
\path[dot_line] (c411)--++(-0.5,-1);

\draw[dot_line, line width=1pt] (-5,-4.2)--(6,-4.2);
\draw[dot_line, line width=1pt] (-5,-7.4)--(6,-7.4);

\node at (5.4,-3.7) {\textbf{1}};
\node at (5.4,-6.8) {\textbf{2}};

\end{tikzpicture}}}
\vfill


\subfloat[]{
\resizebox{0.4\textwidth}{!}{
\begin{tikzpicture}[node distance=1cm, every node/.style={font=\footnotesize}]

\node[blue node] (v351) at (0,0) {$v_{351}$};
\node[red node] (c357) [below=of v351] {$c_{357}$};
\node[red node] (c352) [left=of c357] {$c_{352}$};
\node[red node] (c351) [right=of c357] {$c_{351}$};

\node[violet node] (v478) at (-3.5,-3.5) {$v_{478}$};
\node[violet node] (v483) [below=of c357] {$v_{483}$};
\node[violet node] (v477) at (3.5,-3.5) {$v_{477}$};

\path[line] (v351)--(c352);
\path[line] (v351)--(c357);
\path[line] (v351)--(c351);
\path[line] (v478)--(c352);
\path[line] (v483)--(c357);
\path[line] (v477)--(c351);

\node[black node] (c0) [below left=of v477] {$c_{0}$};
\node[black node] (c405) [below right=of v477] {$c_{405}$};
\node[blue node] (v0) [below=of c0] {$v_{0}$};
\node[blue node] (v405) [below=of c405] {$v_{405}$};

\node[black node] (c6) [below left=of v483] {$c_{6}$};
\node[black node] (c411) [below right=of v483] {$c_{411}$};
\node[blue node] (v0b) [below=of c6] {$v_{0}$};
\node[blue node] (v405b) [below=of c411] {$v_{405}$};

\node[black node] (c1) [below left=of v478] {$c_{1}$};
\node[black node] (c406) [below right=of v478] {$c_{406}$};
\node[blue node] (v0a) [below=of c1] {$v_{0}$};
\node[blue node] (v405a) [below=of c406] {$v_{405}$};

\path[line] (c0)--(v477);
\path[line] (c405)--(v477);
\path[line] (c0)--(v0);
\path[line] (c405)--(v405);

\path[line] (c6)--(v483);
\path[line] (c411)--(v483);
\path[line] (c6)--(v0b);
\path[line] (c411)--(v405b);

\path[line] (c1)--(v478);
\path[line] (c406)--(v478);
\path[line] (c1)--(v0a);
\path[line] (c406)--(v405a);

\path[dot_line] (c352)--++(-1.5,-1);
\path[dot_line] (c352)--++(0.5,-1);
\path[dot_line] (c357)--++(-0.5,-1);
\path[dot_line] (c357)--++(0.5,-1);
\path[dot_line] (c351)--++(1.5,-1);
\path[dot_line] (c351)--++(-0.5,-1);

\draw[dot_line, line width=1pt] (-5,-4.2)--(6,-4.2);
\draw[dot_line, line width=1pt] (-5,-7.4)--(6,-7.4);

\node at (5.4,-3.7) {\textbf{1}};
\node at (5.4,-6.8) {\textbf{2}};

\end{tikzpicture}}}

\caption{Computation tree and the potential structure indicating the scope of separation improvement for all the trapped qubits of one of the disjoint subsets of the $(6,0)$ QTS of $[[882,24]]$ GHP code.}
\label{fig:ct-qts60}
\end{figure}
In Fig. \ref{fig:ct-qts60}, we show three computation trees of the trapped data qubits $v_0$, $v_{405}$ and $v_{351}$ respectively.
All of them belong to the $(6,0)$ QTS.
We observe from the computation trees of $T(v_i) : i \in \{0, 405, 351\}$, that in level $\mathcal{K}=2$ of each CT there are exactly $3(3-1) = 6$ stabilizer check nodes, which are directly adjacent to the data qubits belonging to the same disjoint qubit subset of the QTS, indicated by same color.\\
We now show that removing these check nodes from the computation tree helps the min-sum BP decoder converge and correct the QTS-supported harmful errors.
\begin{lemma}
    Any erroneous qubit $v_r \in V_i$, where $V_i$ is one of the isomorphic degenerate data qubit subsets of the QTS; can be corrected by \textit{min-sum} BP decoder if the $d_v(d_v-1)$ number of check nodes from the level $2$ of $T(v_r)$, which are limiting the separation of $v_r$; are removed from the computation tree of the decoder.
    \label{cor:sep-imp-rem-chk}
\end{lemma}
\begin{proof}
    To prove the above lemma, we analyze the message received by one of the erroneous data qubits in the QTS.
    The computation tree traversed in root-to-leaf order represents the exact sequence of messages passed by the min-sum BP decoder.
    Therefore, each iteration of the decoder indicates a level in the computation tree.
    Now assume $v_r \in V_i$, where $V_i$ is one of the degenerate qubit subsets of the QTS and $c_r \in N(V_i)$ is one of the common stabilizer check nodes, which connects two different disjoint qubit subsets of the QTS.
    Also, $T(v_r)$ is the computation tree with $v_r$ as its root.
    We now observe the flow of information in $T(v_r)$ from leaf nodes towards the root node.
    For the \textit{min-sum} BP decoder to be able to correct erroneous $v_r$, correct information must be passed from $c_r$ to root $v_r$ in $T(v_r)$.
    For this to be done, all the information passed to $c_r$ from its descendants must be correct.
    This correct passage of information must satisfy the following equality for messages passed from any child node $v_i$ to the parent $c_r$,
    \begin{align}
        \mathrm{sign}(\lambda_i) = \mathrm{sign}(m_{v_i \rightarrow c_r}); \forall v_i \in N(c_r)\backslash v_r.
        \label{eq:correct-info-sign}
    \end{align}
    Eq. \eqref{eq:correct-info-sign} implies that, as information passed into $c_r$ is correct, the sign of all the messages passed into $c_r$ should be the same as that of the message passed into data qubit nodes $v_i$ from the channel.
    Now, assuming a phenomenological bit flip noise channel, the magnitude of the messages received by the data qubit nodes from the channel is equal.
    Therefore, for $d_v$ regular QLDPC code, the messages sent from any data qubit $v_i$ in the same level of $T(v_r)$ have equal magnitude and follow
    \begin{align}
        |m_{v_i \rightarrow c_j}| &= |\lambda_i| + \sum_{c_{j^\prime}:N(v_i)\backslash c_j}|m_{c_{j^{\prime}}\rightarrow v_i}|.
    \end{align}
    Therefore, if we consider the effect of the stabilizer check nodes, which belong to the common neighborhood of the disjoint subsets, at level $2$ of $T(v_r)$, we directly observe a violation of Eq. \eqref{eq:correct-info-sign}.
    This is because of the following condition
    \begin{align}
        \mathrm{sign}(m_{c_j \rightarrow v_i}) \neq \mathrm{sign}(\lambda_i),
        \label{eq:violate-info-sign}
    \end{align}
    where if $v_r \in V_i$, then $v_i \in V_j$; such that $V_i \cap V_j = \emptyset$.
    At the junction of level $1$ and level $2$, qubit $v_i$ is adjacent to $d_v - 1$ child check nodes, and all such nodes follow Eq. \eqref{eq:violate-info-sign}.
    Therefore, the check nodes at level $2$ of $T(v_r)$, which limit the separation of $v_r$, also pass incorrect information to check node $c_r$ adjacent to $v_r$, which leads to net incorrect information passed to $v_r$ by the \textit{min-sum} BP decoder.
    Therefore, removal of all such $d_v(d_v - 1)$ number of check nodes from level $2$ of $T(v_r)$ ensures correct information passed to the erroneous qubits and, as a result, ensures correctness of the erroneous qubits in the QTS.
\end{proof}
In table \ref{table:qts-882}, we validate the lemma \ref{cor:sep-imp-rem-chk} and observe that removal of those separation limiting $6$ stabilizer check nodes actively improves the \textit{min-sum} BP decoder's performance and ensures correction of the root erroneous trapped qubit.
\begin{table}[!ht]
    \centering
    {
    \renewcommand{\arraystretch}{1.75}
    \setlength{\tabcolsep}{3pt}
    \begin{tabular}{|c|c|c|}
        \hline
        \parbox[c]{4.3cm}{\centering Removed stabilizer check nodes} & \parbox[c]{1.35cm}{\centering Affected trapped qubit} & \parbox[c]{2.15cm}{\centering New syndrome bit predictions by \textit{min-sum} decoder}\\
        \hline
        \parbox[c]{4.3cm}{\raggedright $[c_{406},c_{352},c_{351},c_{405},c_{357},c_{411}]$} & \parbox[c]{1.35cm}{\centering $v_0$} & \parbox[c]{2.15cm}{\centering $[c_0, c_1, c_6]$}\\
        \hline
        \parbox[c]{4.3cm}{\raggedright $[c_{0},c_{405},c_{1},c_{406},c_{6},c_{411}]$} & \parbox[c]{1.35cm}{\centering $v_{351}$} & \parbox[c]{2.15cm}{\centering $[c_{351},c_{352},c_{357}]$}\\
        \hline
        \parbox[c]{4.3cm}{\raggedright $[c_{0},c_{351},c_{1},c_{352},c_{6},c_{357}]$} & \parbox[c]{1.35cm}{\centering $v_{405}$} & \parbox[c]{2.15cm}{\centering $[c_{405},c_{406},c_{411}]$}\\
        \hline
    \end{tabular}
    }
    \caption{Direct check node removal showing the positive impact on the prediction of erroneous trapped qubits by \textit{min-sum} BP decoder for a $(6,0)$ QTS.}
    \label{table:qts-882}
\end{table}
In the following section, we first discuss the classical results from Zhang \emph{et al.} and elaborate on the limitations imposed by the probabilistic decoding algorithm of check node removal.
Next, we propose incorporating the information measurement method \cite{xu2018iterative} as a subroutine to guide a more targeted removal of stabilizer check nodes, thereby improving the accuracy of the min-sum message-passing decoder.

\section{Improving the min-sum decoder using stabilizer Check Node Removal}
\label{sec:collab-decoding}
As discussed in the previous section, removing the check nodes from the computation tree of an erroneous trapped node increases its separation and aids its correction by the min-sum BP decoding algorithm.
In this section, we describe the complete decoding algorithm and a collaborative decoding pipeline, inspired by \cite{kang2015breaking}, resulting in a low-complexity decoding scheme.
\subsection{Selection of the Computation-Tree Level for Check-Node Removal}
In the general settings, determining the exact level in the computation tree (CT) from which the check nodes should be removed requires prior knowledge of the trapping sets.
For example, from Figs. \ref{fig:computation-tree-3-3} and \ref{fig:ct-qts60}, the level $\mathcal{K}$ of the computation tree containing the removable stabilizer check nodes for the root erroneous trapped qubit can be identified, since the trapping-set configurations are known from Fig. \ref{fig:trapping-sets-of-882-24}.
However, for general families of quantum codes, identifying trapping-set configurations is a nontrivial task. 
Therefore, it is essential to determine the computation-tree level at which candidate check nodes should be selected for removal without relying on prior knowledge of trapping sets.\\
In the computation tree of a trapped data qubit belonging to a quantum trapping set (QTS), level $2$ consistently contains data qubits from the complementary degenerate subset. This observation motivates the choice $\mathcal{K} = 2$ in $T(v_r)$, for $v_r \in V_a, V_b$, as the level from which stabilizer check nodes are selected for removal. 
In contrast, classical trapping sets (CTSs) do not exhibit such a symmetric structure. Consequently, in the absence of structural information about a CTS, it is generally not possible to preselect an appropriate computation-tree level for stabilizer check-node removal.
Kang \emph{et al.} proposed the concept of the weak check nodes, which relaxes the condition of removing the exact check node, which is responsible for limiting the separation of a trapped bit \cite{kang2015breaking}.
This weak check node in the computation tree of a classical bit refers to a check node in any level of the CT having the shortest distance to any of the trapped bits as a descendant of the root bit.
Therefore, choosing any level $\mathcal{K}$ of the computation tree and probabilistic removal of check nodes from the same level results in a separation extension of the trapped root bit and therefore allows the message passing BP decoder to correct the root erroneous bit.\\
By adopting the same approach for CTSs, stabilizer check nodes can be removed from level $\mathcal{K} = 2$ of the computation tree of any trapped qubit to increase the separation of the trapped qubits in any CTS. 
This provides a practical strategy for selecting the computation-tree level for stabilizer check-node removal in both CTS and QTS scenarios without requiring prior knowledge of their configurations.
During decoding, the decoder has no prior knowledge of the trapped qubits and relies solely on the syndrome, i.e., the violated stabilizer check nodes. Accordingly, we construct computation trees rooted at stabilizer check nodes. In this setting, stabilizer check nodes are removed from level $\mathcal{K} = 1$ of the computation tree $T(c_r)$, where $c_r$ is a violated stabilizer check node.
These observations highlight that the distinct structural characteristics of CTSs and QTSs necessitate different strategies for identifying harmful stabilizer check nodes.
Although we have argued that stabilizer check nodes can be removed from level $\mathcal{K} = 1$ of the computation tree rooted at a violated check node, the questions of how many such nodes should be removed and which specific nodes to select remain open. 
In the following section, we address these questions.

\subsection{Information-Measurement (IM) assisted stabilizer check node selection}
For classical trapping sets, Kang et al. \cite{kang2015breaking} proposed a probabilistic removal of check nodes that improves trapped-qubit separation with or without requiring prior knowledge of the underlying trapping-set configuration. Their approach removes check nodes either through random selection or by subjecting them to a Bernoulli trial.
For large QLDPC codes, a random, unassisted check node removal does not improve the min-sum decoder's performance.
We address this issue by incorporating an information-measurement subroutine that enables more selective identification of stabilizer check nodes that contribute to harmful error beliefs or limit the separation of the trapped qubits.
For each data and stabilizer check qubits we associate an information measurement (IM) value with it.
\begin{definition}[Information Measurement for data qubits]
    For each data qubit of the quantum code, the information measurement is defined by the total number of adjacent stabilizer checks that are unsatisfied. 
\end{definition}
i.e. 
\begin{align}
    \mathrm{IM}_{v_i} = \sum_{j}c_j, \forall c_j \in N(v_i), \text{if}\,\, c_j = 1.
    \label{eq:im-qubit}
\end{align}
\begin{definition}[Information Measurement for Stabilizer checks]
    For each stabilizer check, the information measurement is defined as the sum of the information measurement values of all its adjacent data qubits. 
\end{definition}
i.e.
\begin{align}
    \mathrm{IM}_{c_j} = \sum_{i}\mathrm{IM}_{v_i}, \forall v_i \in N(c_j).
    \label{eq:im-check}
\end{align}
We claim that the stabilizer check nodes that contribute to harmful error beliefs in any QTS have the highest IM values among those that contribute proper error beliefs.
We previously discussed how the common neighborhood of the degenerate subsets of a QTS limits the separation of the trapped qubits and, therefore, they are also responsible for passing incorrect error beliefs.
\begin{lemma}
    The stabilizer check nodes from level $\mathcal{K} = 2$ of the computation tree of an erroneous data qubit $v_r$ that belongs to a QTS, i.e., $v_r \in V_a \cup V_b$; always have the maximum IM values.
\end{lemma}
\begin{proof}
    The symmetry of the degenerate qubit subsets implies that any error configuration supported on the quantum trapping set will always lead to the same non-zero IM values for the data qubits belonging to the two degenerate subsets.
    For instance, let us assume the isolation assumption, where only an error pattern supported on the QTS exists.
    Also assume that for a $(a,0)$ QTS with critical number $\frac{a}{2}$, the error is of cardinality $\frac{a}{2}$ and is supported on one of the degenerate qubit subsets.
    In this scenario, all the check nodes from the common neighborhood of the degenerate subsets $V_a$ and $V_b$ are violated, i.e.
    \begin{align*}
        c_r = 1, \forall c_r \in N(V_a) = N(V_b).
    \end{align*}
    Further, in any QTS, there are no odd-degree stabilizer check nodes. For a $(d_v, d_c)$- regular QLDPC code, all $d_v$ check nodes adjacent to a trapped data qubit belong to the QTS.
    Therefore, from Eq. \ref{eq:im-qubit}, we get
    \begin{align*}
        \mathrm{IM}_{v_r} &= d_v, \forall v_r \in V_a \cup V_b.\\
        \mathrm{IM}_{v_{r^{\prime}}} &= 0, \forall v_{r^{\prime}} \notin V_a \cup V_b.
    \end{align*}
    The above equations along with the fact that any $c_r \in N(V_a) = N(V_b)$ has even degree directly imply
    \begin{align*}
        \mathrm{IM}_{c_r} &= 2d_v, \forall c_r \in N(V_a) = N(V_b),\\
        \mathrm{IM}_{c_{r^{\prime}}} &= 0, \forall c_{r^{\prime}} \notin N(V_a) = N(V_b).
    \end{align*}
    In level $\mathcal{K}=2$ of $T(v_r): v_r \in V_a \cup V_b$, there are always two sets of check nodes, $c_r \in N(V_a) = N(V_b)$ and $c_{r^{\prime}} \notin N(V_a) = N(V_b)$.
    Therefore, the above equations establish the fact that the stabilizer check nodes belonging to the common neighborhood of the trapped data qubits are the only check nodes with non-zero IM values with equal magnitude, and all the other nodes have zero IM value.
\end{proof}
The above lemma indicates that the stabilizer check qubits from the common neighborhood of the degenerate subsets of the QTS have bigger IM values compared to any other stabilizer check qubits at $\mathcal{K} = 2$ of the computation tree with a trapped data qubit at the root.
Therefore, estimating the IM values improves the selectivity for removing stabilizer check nodes, leading to the accurate removal of harmful error beliefs and improving the separation of the trapped data qubits.
\subsection{A collaborative decoding pipeline}
In this section, we propose a collaborative decoding architecture that integrates the IM-assisted stabilizer check node removal from the computation tree of the min-sum decoder to mitigate trapping-set effects and improve the performance of the min-sum BP decoder.
We summarize a minimal version of classical check node removal algorithm, inspired by the work of Kang \emph{et al.} \cite{kang2015breaking} in Appendix \ref{ap:cnr}.
From the algorithm itself, we can observe that no measures have been taken to improve the identification of stabilizer nodes that contribute to harmful error beliefs and limit the separation of the trapped qubits.
We discussed in the previous section that this can be addressed by incorporating IM-based selection of stabilizer check nodes.
The function Find\_IMs evaluates the IM values for an input set of stabilizer check nodes.
Algorithm \ref{alg:find_im_opt} from Appendix \ref{ap:findim} describes the Find\_IMs function.\\
Using Algorithm \ref{alg:find_im_opt}, we propose Quantum Check Node Removal (QCNR) in Algorithm \ref{alg:qcnr}.
\begin{algorithm}
\caption{Quantum Check Node Removal (QCNR)}\label{alg:qcnr}
\SetKwFunction{KwFn}{Find\_IMs}
\SetKwInOut{Input}{input}\SetKwInOut{Output}{output}
\Input{$\mathbf{H}$, $\mathtt{UNSAT} $, $\mathtt{t}$, $\mathtt{df}$}
\Output{Modified parity check matrix $\mathbf{H}_{FN}$}
$\mathbf{H}_{FN} = \mathbf{H}$\tcp*[r]{Initialization}
$\mathtt{rem}$ = []\tcp*[r]{Initialize checks removal}
$\mathtt{im\_checks} = \KwFn(\mathtt{UNSAT})$\;
\For{$\mathtt{check}$ in $\mathtt{UNSAT}$}{
    $\mathtt{leaf}$ $\gets$ All the leaf nodes of $T_t(\mathtt{check})$\;
    $\mathtt{leafval} \gets \left[\, \mathtt{im\_checks}[\ell] \;\middle|\; \ell \in \mathtt{leaf} \,\right]$\;
    $\mathtt{maxval} \gets \max(\mathtt{leafval})$\tcp*[r]{Maximum IM value}
    
    $\mathtt{maxleaf} \gets \{\, \ell \in \mathtt{leaf} \;|\; \mathtt{im\_checks}[\ell] = \mathtt{maxval} \,\}$\tcp*[r]{All check nodes attaining the maximum}
    
    $\mathtt{rem} \gets \mathtt{rem} \cup \mathtt{maxleaf}$
}
$\mathtt{dis}$ $\gets$ random.chose($\mathtt{rem}$,$\mathtt{df}$)\;
\ForEach{$c \in \mathtt{dis}$}{
    $\mathbf{H}_{FN} \gets \mathbf{H}_{FN}[\{1,\dots,m\} \setminus \{c\},\, :]$\tcp*[r]{Remove row $c$}
}
return $\mathbf{H_{FN}}$ \tcp*[r]{The modified matrix}
\end{algorithm}
We further use this QCNR in a collaborative setting to deduce a low complexity decoding architecture.
In the collaborative architecture, we use the min-sum decoder in two modes.
The main mode implements the \textit{min-sum} BP on the unmodified parity check matrix of the QLDPC code.
The other mode, which we call sub-decoding mode, carries out the min-sum algorithm on a modified parity check matrix, which is obtained by removing some stabilizer check nodes from the computation tree of the unsatisfied stabilizer check nodes, using Algorithm \ref{alg:qcnr}.
In previous sections we have seen that the computation tree embodies the message passing iterations of the min-sum decoder and therefore running the same decoder over a modified parity check matrix with certain rows removed essentially implements the message passing rules with the corresponding check nodes ignored.

We have also seen that the main decoding mode is limited in its performance.
Therefore, if, after a predetermined number of message-passing iterations of the min-sum algorithm, the decoder makes no further correct predictions, it switches to sub-decoding mode.
After the sub-decoding round completes, the decoder switches back to the main decoding mode.
We note that in both the main and sub-decoding rounds, we continuously update the list of net-unsatisfied stabilizer check nodes.
Therefore, calling the main decoding mode again after the completion of a sub-decoding mode is aimed at correcting any error that has not been addressed by the removal of stabilizer check nodes of the previous sub-decoding round, and also to tackle errors that are no longer supported over any TS after the previous main and sub-decoding round.
This approach is different than the post-processing based methods, as it switches its message-passing operations over two different parity check matrices and outputs a prediction only if all the errors are corrected, or a predetermined number of sub-decoding rounds have been completed.
\begin{algorithm}
\caption{Quantum Collaborative Check Node Removal (QCCNR) Algorithm}\label{alg:bp-qcnr}
\SetKwInOut{Input}{input}\SetKwInOut{Output}{output}
\Input{\begin{itemize}
    \item Syndrome:= $\mathtt{s}$ and pcm $\mathbf{H}$,
    \item Channel llr values:= $\mathtt{p}$,
    \item Maximum iterations of main BP:= $\mathtt{max_{iter}}$,
    \item Maximum iterations of sub BP:= $\mathtt{max_{sub}}$,
    \item Deselection degree: $\mathtt{df}$,
    \item Maximum sub-decoding rounds: $\mathtt{ns}$,
    \item Tolerance of BP iterations: $\mathtt{tol}$.
\end{itemize}}
\Output{Predicted error pattern $\mathtt{\hat{e}}$.}
\BlankLine\tcp{Initialization with the Main mode}
$\mathtt{\hat{e}_{bp}}, \mathtt{\hat{s}_{bp}} = \mathsf{MAIN\_BP}(\mathbf{H}, p, \mathtt{max_{iter}}, \mathtt{tol}, \mathtt{s})$\;
\BlankLine
\tcp{Sub-decoder initialization}
$\mathtt{e_{net}}:= \mathsf{zeros}\mathtt{((}\mathtt{ns},\mathbf{H}.\mathsf{shape}\mathtt{[1]}\mathtt{))}$\;
$\mathtt{s_{net}}:= \mathsf{zeros}\mathtt{((}\mathtt{ns},\mathbf{H}.\mathsf{shape}\mathtt{[0]}\mathtt{))}$\;
\While{$\mathtt{i < ns}$}{
   $\mathtt{s_{fn} \gets s + \hat{s}_{bp} + (\sum_{j}s_{net}[j]), \forall j < i}$\;
    $\mathtt{UNSAT} = \mathrm{supp}(s_{fn})$\;
        
   $\mathbf{H}_{FN} = \mathsf{QCNR}(\mathbf{H},\mathtt{UNSAT},\mathtt{t} = 1,\mathtt{df})$\;
   \BlankLine\tcp{Sub-decoding mode}
   $\mathtt{\hat{e}_{sub}}, \mathtt{s_{sub}} = \mathsf{SUB\_BP}(\mathbf{H}_{FN}, p, \mathtt{max_{sub}}, \mathtt{tol}, \mathtt{s_{fn})})$\;
    \BlankLine\tcp{Switching to Main decoding mode}
    $\mathtt{\hat{e}^{\prime}_{bp}}, \mathtt{\hat{s}^{\prime}_{bp}} = \mathsf{MAIN\_BP}(\mathbf{H}, p, \mathtt{max_{iter}}, \mathtt{tol}, \mathtt{s_{sub}} + \mathtt{s_{fn}})$\;
    $\mathtt{e_{net}}[i] \gets \mathtt{e_{sub}} + \mathtt{\hat{e}^{\prime}_{bp}}$\;
    $\mathtt{s_{net}}[i] \gets \mathtt{s_{sub}} + \mathtt{\hat{s}^{\prime}_{bp}}$\; 
   \If{$\mathrm{supp}(\mathtt{s_{fn}} + \mathtt{s_{net}}[i]) = \emptyset$}{
    return $\mathtt{\hat{e} \gets \hat{e}_{bp} + \sum_{m: m \leq i}\mathtt{e_{net}}[m]}$\;
    \textbf{End}\;
   }
   $\mathtt{i} \gets \mathtt{i} + \mathtt{1}$
}
return $\mathtt{\hat{e} \gets \hat{e}_{bp} + \sum_{i}\hat{e}_{net}[i]}$    
\end{algorithm}
A detailed set of steps for this collaborative decoding is described in Algorithm \ref{alg:bp-qcnr}.\\\\
In Fig. \ref{fig:break-stuck-bp}, we show the advantage of using the min-sum algorithm in the \textit{sub}-decoding mode.
We observe, out of numerous sampled instances, a case where the min-sum algorithm gets stuck while decoding the syndrome of an error pattern generated at an independent bit flip physical error rate of $p = 0.03$ for the $[[882,24]]$ GHP code after a few iterations.
We note that for all the numerical experiments of GHP codes, we use the following parameter values for the BP decoding: maximum-BP-iterations = $100$, BP-method = ``minimum-sum'', scaling-factor = $0.625$ \cite{emran2014simplified}.
These parameters are relevant to the message passing of the \textit{min-sum} BP in both the main and sub-decoding modes.
The stuck scenario indicates that the main decoder cannot identify any additional correct stabilizer check violations.
After the main mode of BP decoder fails to correct any more errors for a consecutive $11$ rounds, we iteratively apply Algorithm \ref{alg:qcnr} and run the BP decoder on a modified parity check matrix.
We keep track of the new corrections at each round of the sub-decoding.
The tail of the plot indicates that by activating the \textit{sub}-decoding mode, new stabilizer violations are predicted correctly and thus, eventually, all the errors are corrected.
This numerically confirms that removing stabilizer check nodes using Algorithm \ref{alg:qcnr} increases the efficacy of the decoding output.
\begin{figure}[t]
  \raggedright
  \includegraphics[width=0.45\textwidth]{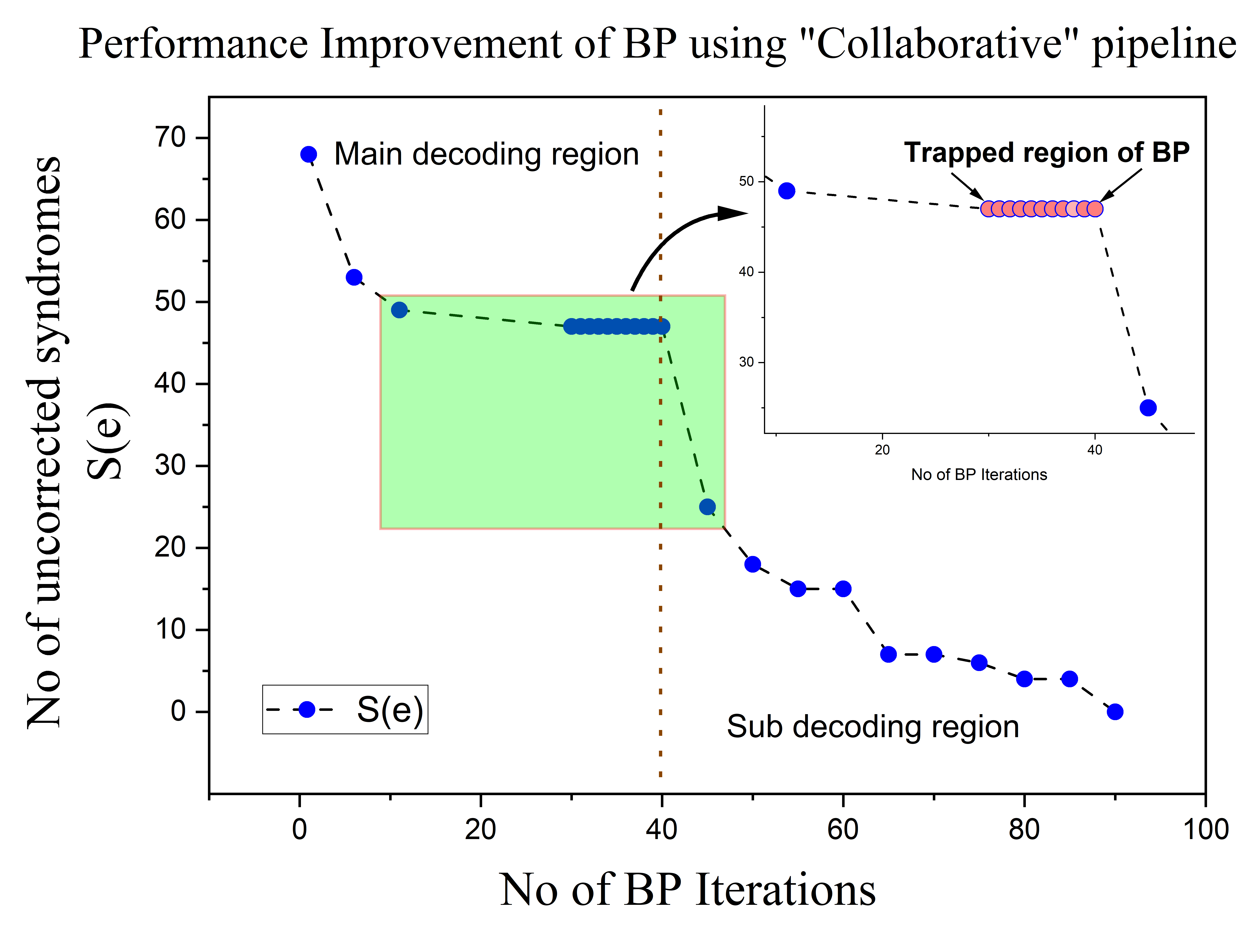}

  \caption{Numerical experiment showing the \textit{sub}-decoder's ability to identify stabilizer violations correctly, which were not possible by decoding with \textit{min-sum} based BP alone (which we denote as `main decoding region').
  We assume that the error occurs on the data qubits after each round of an error-correction cycle, and that the syndrome measurements are perfect.
  Under this code-capacity noise model, we sample an error pattern at an independent bit flip physical error rate of $p = 0.03$ for the $[[882,24]]$ GHP code.
  The \textit{min-sum} BP decoder cannot predict any more correct stabilizer violations for $11$ consecutive iterations and gets stuck at the $40^{th}$ decoding iteration.
  Subsequently, we initialize the sub-decoding mode.
  We call Algorithm \ref{alg:qcnr} repeatedly, and \textit{min-sum} BP is applied on the output parity check matrix of the QCNR algorithm.
  To tackle all the trapped cases for QTS and CTS, we set the deselection degree to $\mathtt{df} = 6$ for the first few rounds and then set $\mathtt{df} = 1$ for the last $20$ sub-decoding rounds.
  We observe repeated calls of Algorithm \ref{alg:qcnr} leading to an improvement in the performance of the min-sum algorithm in the \textit{sub}-decoding region. 
  The tail of the plot signifies newly accurate predictions of the violated stabilizer checks, indicating the improvement in performance.}
  \label{fig:break-stuck-bp}
\end{figure}
\\\\
We perform memory experiments to evaluate the performance of the proposed QCCNR decoder in Algorithm \ref{alg:bp-qcnr}, comparing it with the \textit{min-sum} BP decoder and the state-of-the-art BP+OSD decoder with OSD order set to $0$. In these experiments, errors on the GHP code are decoded under a code-capacity bit-flip noise model. The maximum number of sub-decoding rounds is set to $200$, and a halting condition is imposed such that the decoder terminates once all violated syndromes are corrected, or the maximum number of sub-decoding rounds is reached, as specified in Algorithm \ref{alg:bp-qcnr}.
Now, we have adopted a scenario where in the first $100$ sub-decoding rounds we set the deselection degree to $\mathtt{df} = 6$, and for the rest of $100$ rounds we set $\mathtt{df} = 1$.
This strategy is adopted deliberately to first target the errors supported on any QTS and then to target the errors supported on the CTS.
The GHP code family used in the experiments has a data qubit degree $d_v = 3$.
Therefore, from Lemma \ref{cor:cor-no-check-rem}, we set the deselection degree to $\mathtt{df} = d_v(d_v - 1) = 6$ to increase the separation of the trapped qubits inside the QTS.
In Fig \ref{fig:decoder-comp-ghp}, we show the results, which indicate that the proposed decoding scheme offers significant improvements in the logical error rates over the min-sum BP decoder, indicating a potential breakthrough to circumvent the harmful configurations for the iterative min-sum BP.
Algorithm \ref{alg:bp-qcnr} gives logical error rates much better than standard min-sum BP and almost achieves a performance comparable to the OSD$0$ post-processor.
\begin{figure}[ht!]
  \centering
  \subfloat[]{  \includegraphics[width=0.45\textwidth]{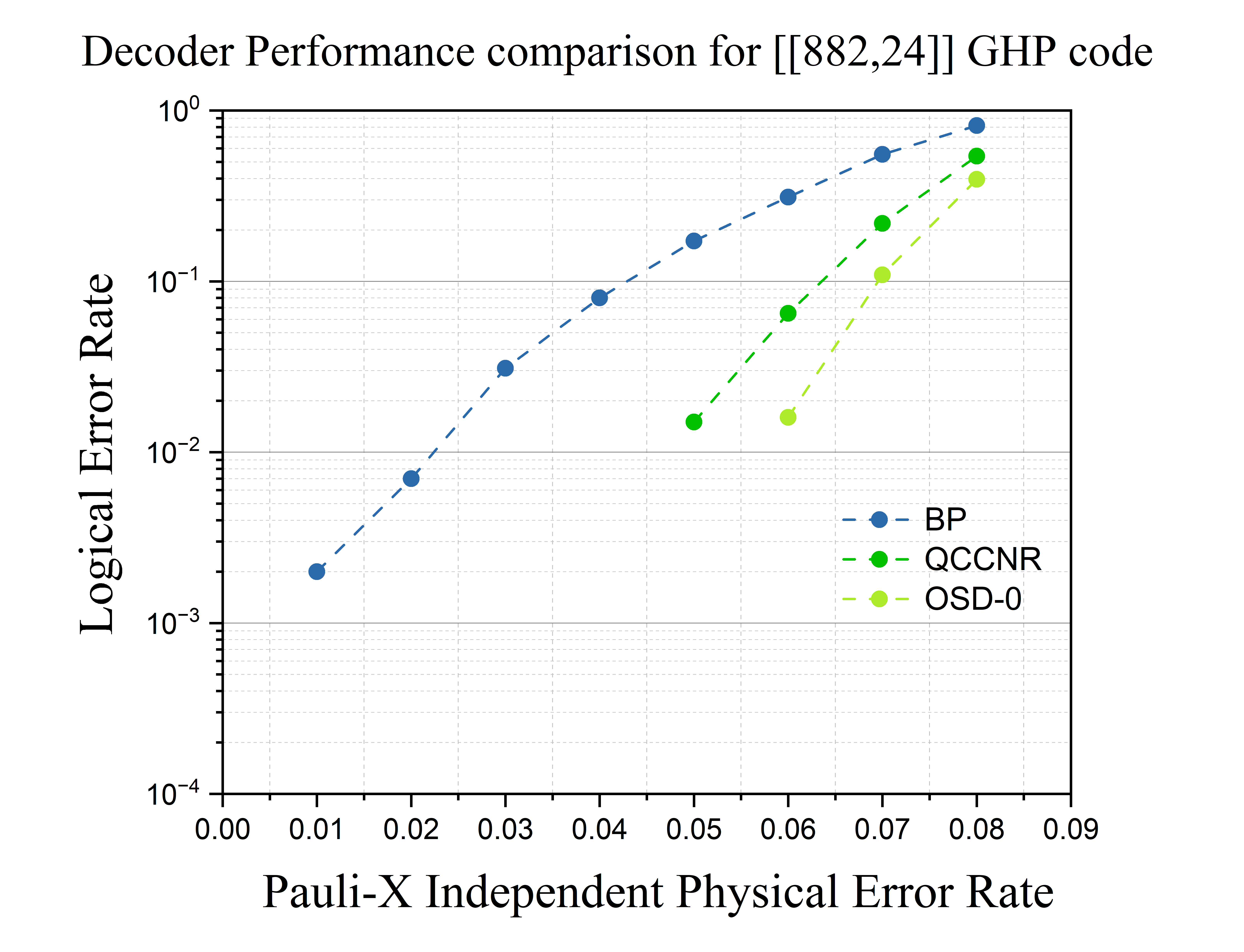}
      \label{fig:882-decoder-comp}
      }
    \hfill
    \subfloat[]{\includegraphics[width=0.45\textwidth]{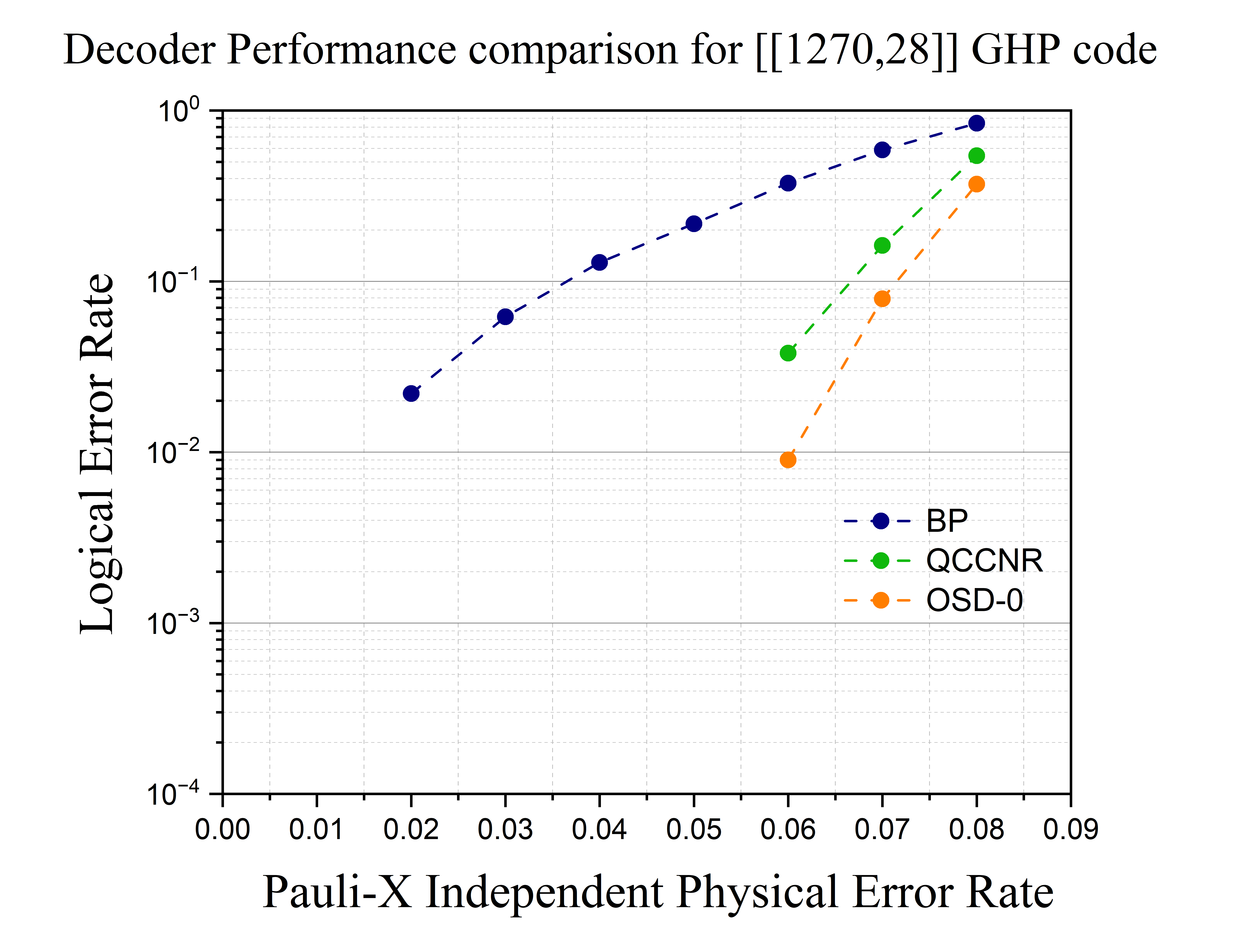}
      \label{fig:1270-decoder-comp}
      }
  \caption{The logical error rates obtained for the GHP codes from Ref. \cite{panteleev2021degenerate}, under the independent Pauli-$\mathrm{X}$ noise channel.
  We perform $10^{6}$ Monte Carlo code capacity LER simulations for both cases.
  In (a) we benchmark with the $[[882,24]]$ GHP code and in (b) we benchmark the $[[1270,28]]$ GHP code.
  In both the experiments, we use the \textit{min-sum} algorithm of BP, and for benchmarking against the OSD decoder, we use the the $0^{\text{th}}$ order OSD, i.e., BP+OSD$0$.
  For the QCCNR algorithm, we use the following parameters for Algorithm \ref{alg:bp-qcnr}: $\mathtt{max_{iter}} = \mathtt{max_{sub}} = 100, \mathtt{ns} = 200$.
  Also, for the first $100$ sub-decoding rounds, the deselection degree is set as $\mathtt{df} = 6$.
  This is set from Lemma \ref{cor:sep-imp-rem-chk} to support qubits trapped inside QTS.
  Further, for the remaining $100$ sub-decoding rounds, the deselection degree is set to $\mathtt{df} = 1$, in support of the qubits trapped inside CTS.
  We observe that the proposed QCCNR algorithm provides OSD-like improvements and obviously is far superior to the standard \textit{min-sum} algorithm-based BP decoder.}
  \label{fig:decoder-comp-ghp}
\end{figure}
\subsection{Time complexity}
The collaborative decoder in Algorithm \ref{alg:bp-qcnr} has two modes, as discussed previously.
In both main and sub-decoding mode, the message passing of the min-sum BP decoder is linear, i.e., a $\Theta(n)$ process.
The sub-decoding mode, i.e. Algorithm \ref{alg:qcnr} has two parts: one is the construction of computation trees for each of the unsatisfied check nodes, and the other is finding the IM values of the stabilizer checks.
Given a QLDPC code, we compute the graph adjacency once and it is a $\approx\Theta(n^2)$ process.
The output of the graph adjacency is a list $\mathtt{nc}$ storing the neighborhood of each check node and another list $\mathtt{nv}$ for the neighborhood of the data qubits.
The adjacency lists are further used in Algorithm \ref{alg:find_im_opt} and Algorithm \ref{alg:layer1_tree}.\\\\
Algorithm \ref{alg:find_im_opt}, which computes the IM values has a worst case time complexity of $\Theta(||\mathtt{UNSAT}||)$.
Essentially, for each round of the sub-decoding, we only require the IM values for a constant number of leaf check nodes, in the computation tree of each of the remaining unsatisfied check nodes.
For GHP codes this time amounts to $\Theta(n\log{n})$ \cite{panteleev2021quantum}.
The other process involved in Algorithm \ref{alg:qcnr} is the construction of a level $1$ computation tree for each of the remaining unsatisfied check nodes.
This process, shown in Algorithm \ref{alg:layer1_tree} can be performed in constant time $\Theta(d_cd_v)$, where the underlying QLDPC is $(d_v,d_c)$ regular.
Also $d_v, d_c << n$ and are constants for a given family of QLDPC code.
Therefore, the overall time cost of the decoder is $\approx \Theta(n^2)$, which suggests that QCCNR stands as an excellent strategy for improving the iterative decoder's performance compared to decoders like BP+OSD0 or the higher-order alternatives.

\section{Conclusion and Future Work}
In this work, we propose a collaborative decoding approach for improving the iterative min-sum decoding of QLDPC codes through improvements in \textit{qubit separation} for the trapped qubits, with a special focus on GHP codes.
We propose a new way to measure the separations of the qubits trapped inside a symmetric stabilizer set.
We used this method to go beyond the analysis of qubit separation of the classical trapping sets.
Our decoding architecture is free of any post-processing method and can be considered as a two-mode decoder with the message passing operations of the min-sum algorithm switches over the complete and a modified Tanner graph (or parity check matrix) of the quantum code.
The overall decoding takes approximately $\Theta(n^2)$ time.
Despite being a post-processing-free decoding, the success of the algorithm has a significant improvement compared to the standard iterative \textit{min-sum} decoder.\\\
We would like to conjecture that the the current QCNR algorithm, which selectively removes stabilizer check nodes to improve the trapped qubit separations can not address point-like defects in codes like surface code.
Further we have discussed that the check node removal can prevent harmful error beliefs from passing while the message-passing iterations are going on.
Although we observed cases where certain check nodes, when removed can have a negative effect on a long error chain covering data qubit nodes contributed from both the circulant matrices used in the construction of the GHP code.
The second round of main decoding mode assists in mitigating such negative effects but not all.
These cumulatively prevents our decoder to achieve the OSD-$0$ success rates.
The improved \textit{qubit separation} can assist only in breaking the trapping set scenarios and is incapable of assisting the iterative decoders when point-like syndrome or long error chains with the above particular structure occurs.

In future work, we will investigate improvements to the QCCNR decoder to address these limitations. 
In particular, it will be interesting to examine whether recent studies on optimized message-passing schedules \cite{casado2007informed, moradi2026sequential}, combined with stabilizer check-node removal, can mitigate erroneous predictions arising in the sub-decoding mode.
We also plan to explore ensemble decoding strategies \cite{muller2025improved, koutsioumpas2025automorphism} in the context of QCCNR, to provide a systematic way to address the lack of prior knowledge of optimal parameter choices for irregular QLDPC codes.

To summarize, despite its simplicity, the proposed decoding architecture demonstrates significant performance improvements over the standard min-sum BP algorithm while maintaining low computational overhead. 

\section*{Data Availability}
The data and scripts used for the numerical experiments are available from the authors on reasonable requests.

\begin{acknowledgments}
M.B. acknowledges the use of the \href{https://github.com/quantumgizmos/ldpc_v2}{ldpc-repository} developed by J. Roffe \cite{Roffe_LDPC_Python_tools_2022}, which provided the BP and BP-OSD decoders \cite{roffe_decoding_2020} employed in the memory experiments. 
M.B. and A.R. thank IISER Bhopal for providing the computational resources required for the exhaustive simulations. 
M.B. acknowledges support from the doctoral research fellowship of IISER Bhopal.
A.R. acknowledges funding support from the Department of Science and Technology (DST), Government of India, under the National Quantum Mission (NQM), DST/QTC/NQM/QComm/2024/2.
A.R. is thankful for the grant received from the U.S.—India Science and Technology Endowment Fund (USISTEF), USISTEF/QT/165/2023. 
\end{acknowledgments}

\appendix

\section{GHP and GB codes used for analysis}
\label{ap:codes}
The GHP and GB codes used during the simulations and trapping set-related analysis are from Ref. \cite{panteleev2021degenerate}.
The GHP code construction follows Eq. \ref{eq:pcm-ghp-code}.
The matrices $A \in \mathcal{M}_n(R)$ are given in a polynomial form.
Each polynomial represents a unique $L \times L$ circulant matrix.
These details are as follows \cite{panteleev2021degenerate}:
\begin{enumerate}
    \item $[[882,24]]$ GHP code:
        \begin{align}
            \label{eq:882-24-ghp}
            A &= \begin{pmatrix}
                x^{27} & 0 & 0 & 0 & 0 & 1 & x^{54}\\
                x^{54} & x^{27} & 0 & 0 & 0 & 0 & 1\\
                1 & x^{54} & x^{27} & 0 & 0 & 0 & 0\\
                0 & 1 & x^{54} & x^{27} & 0 & 0 & 0\\
                0 & 0 & 1 & x^{54} & x^{27} & 0 & 0\\
                0 & 0 & 0 & 1 & x^{54} & x^{27} & 0\\
                0 & 0 & 0 & 0 & 1 & x^{54} & x^{27}
            \end{pmatrix},\\
            bI_m &= \begin{pmatrix}
                1 + x + x^6
            \end{pmatrix}I_7,\\
            L &= 63,\\
            (d_v, d_c) & = (3,6)
        \end{align}
    \item $[[1270,28]]$ GHP code:
        \begin{align}
            \label{eq:1270-28-ghp}
            A &= \begin{pmatrix}
                1 & 0 & x^{51} & x^{52} & 0\\
                0 & 1 & 0 & x^{111} & x^{20}\\
                1 & 0 & x^{98} & 0 & x^{122}\\
                1 & x^{80} & 0 & x^{119} & 0\\
                0 & 1 & x^5 & 0 & x^{106}
            \end{pmatrix},\\
            bI_m &= \begin{pmatrix}
                1 + x + x^7
            \end{pmatrix}I_5,\\
            L &= 127,\\
            (d_v, d_c) & = (3,6)
        \end{align}
\end{enumerate}

The GB code construction follows Eq. \ref{eq:pcm-gb-code}.
Elaborate detailing of the GB codes used in the threshold estimation is sourced from Ref. \cite{J_D_Crest_qldpc_codes} and as follows:
\begin{enumerate}
    \item $[[126,12,d<11]]$ GB code:
    \begin{align}
        a &= \begin{pmatrix}
                1 + x^{43} + x^{37}
            \end{pmatrix},\\
        b &= \begin{pmatrix}
                1 + x^{59} + x^{31}
            \end{pmatrix},\\
            L &= 63,
    \end{align}
    \item $[[254,14,d<17]]$ GB code:
    \begin{align}
        a &= \begin{pmatrix}
                1 + x^{18} + x^{53}
            \end{pmatrix},\\
        b &= \begin{pmatrix}
                1 + x^{12} + x^{125}
            \end{pmatrix},\\
            L &= 127,
    \end{align}
    \item $[[510,16,d<19]]$ GB code:
    \begin{align}
        a &= \begin{pmatrix}
                1 + x^{250} + x^{133}
            \end{pmatrix},\\
        b &= \begin{pmatrix}
                1 + x^{41} + x^{157}
            \end{pmatrix},\\
            L &= 255,
    \end{align}
\end{enumerate}
\section{Threshold of Generalized Bicycle codes}
In this appendix, we report the observed code-capacity thresholds for the Generalized Bicycle (GB) codes from Appendix C of \cite{panteleev2021degenerate} under the depolarizing noise channel, using both the QCCNR and BP+OSD0 decoders. 
We consider the $[[126,12,d<11]]$, $[[254,14,d<17]]$, and $[[510,16,d<19]]$ codes in the memory experiments.
\begin{figure}[ht!]
    \centering
  \subfloat[]{  \includegraphics[width=0.45\textwidth]{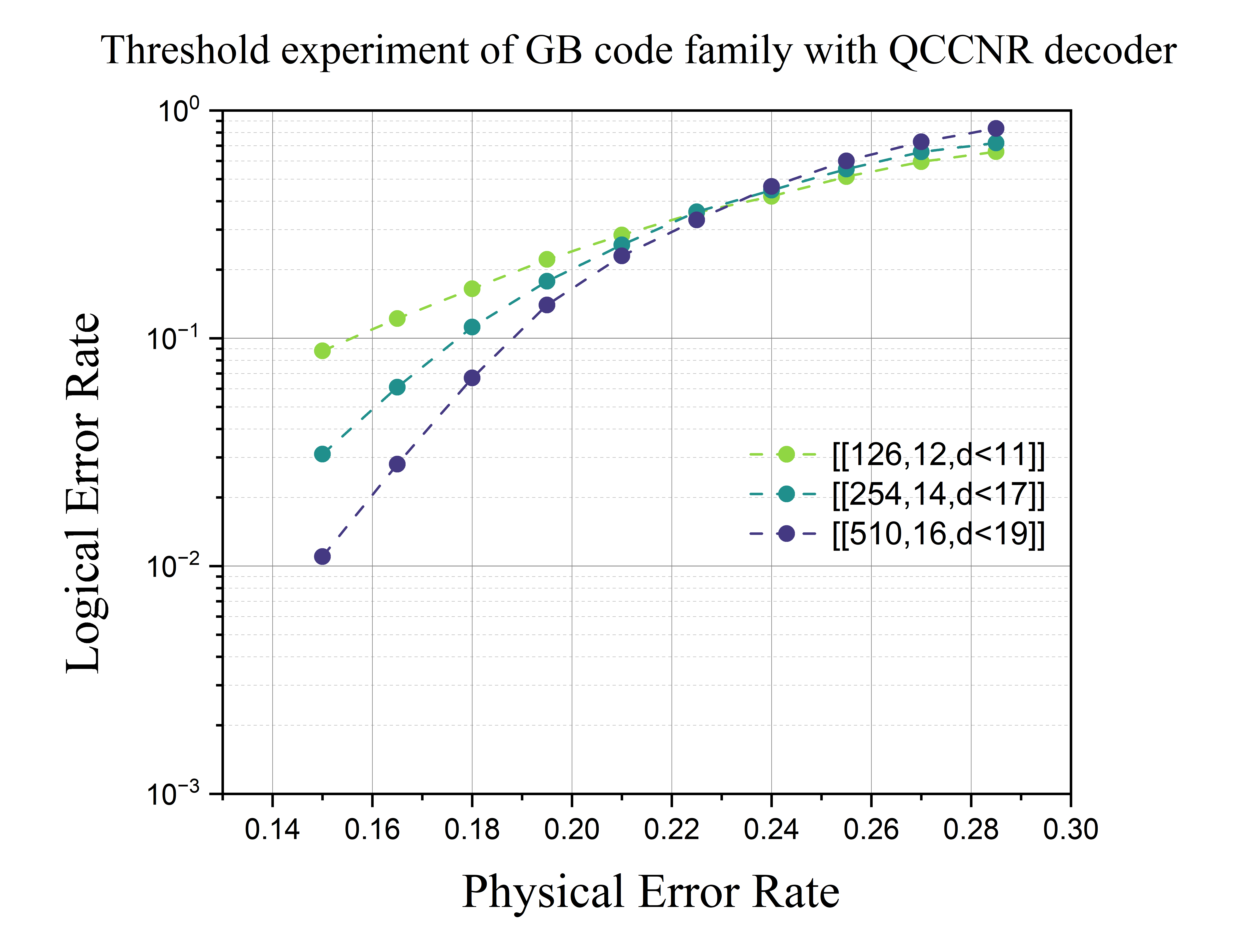}
          \label{fig:threshold-depo-gb-qccnr}
      }
    \hfil   
      \subfloat[]{\includegraphics[width=0.45\textwidth]{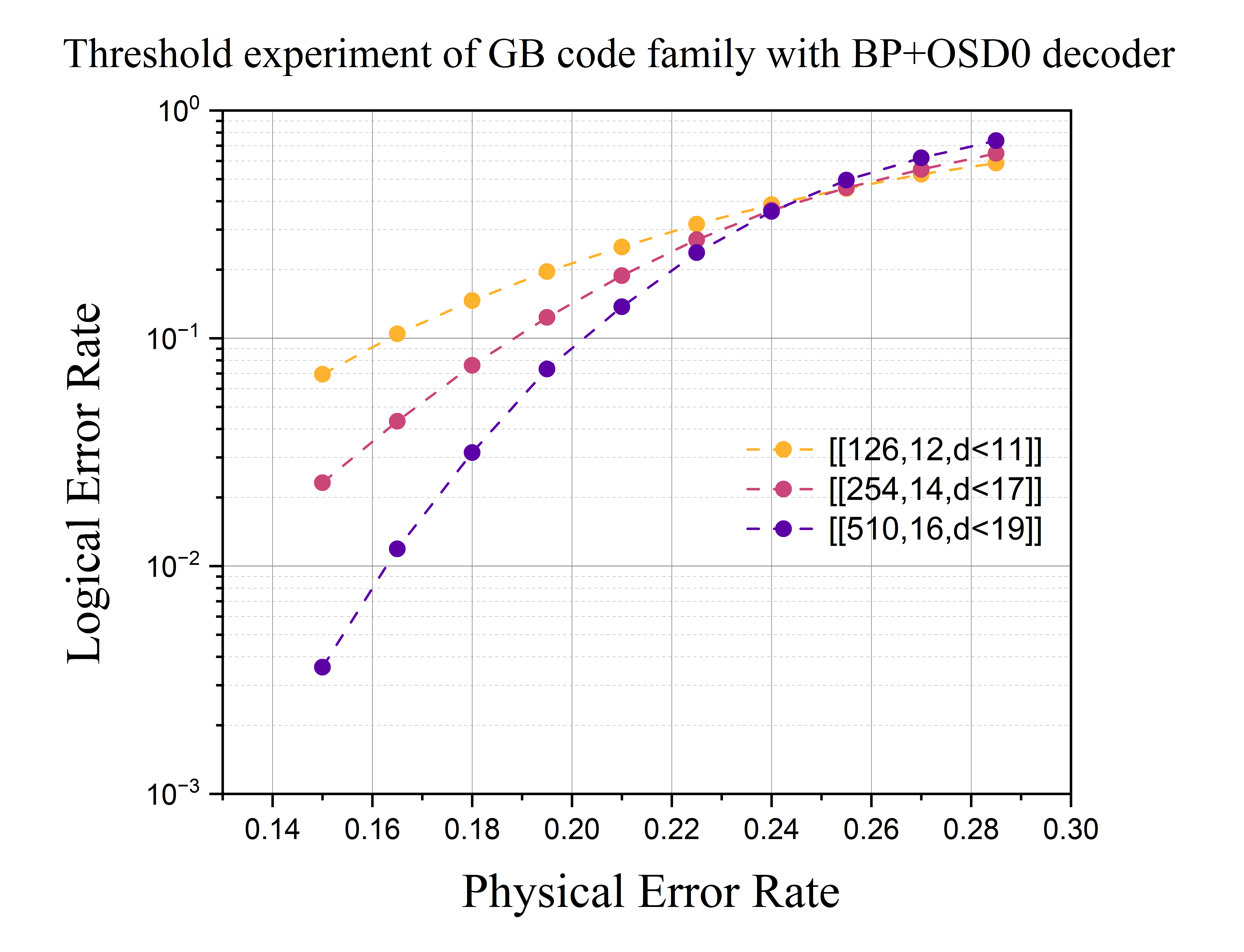}
          \label{fig:threshold-depo-gb-osd}
      }
    \caption{Code capacity threshold obtained for the GB codes under the depolarizing noise channel.
    In (a), we use the QCCNR decoder, and in (b), we use the BP+OSD$0$ decoder for decoding purposes.
    The observed thresholds for QCCNR and OSD$0$ are around $23\%$ and $24\%$, respectively.
    We observe a slightly reduced threshold and increased LERs for the former case.
    Although the results indicate the resolution of trapping sets scenarios under the decoding parameters used.
    }
    \label{fig:threshold-gb-depo}
\end{figure}
We observe a code-capacity threshold of $23\%$ under our proposed decoder.
\section{Probabilistic Check Node Removal algorithm}
\label{ap:cnr}
Algorithm \ref{alg:cnr} describes a stabilizer check-node removal procedure inspired by the work of Kang \emph{et al.} \cite{kang2015breaking}. 
The algorithm begins with the construction of computation trees rooted at each of the unsatisfied stabilizer check nodes. 
Subsequently, a predetermined number of stabilizer check nodes are selected at random from the set of all the leaf stabilizer check nodes at level $t$ of each computation tree.
This number is called the deselection degree $\mathtt{df}$ and is an essential tuning parameter to improve the separation of the trapped qubits from different types of trapping sets.
The CNR algorithm outputs a modified parity check matrix of the quantum code with some rows removed, which corresponds to the removed stabilizer check nodes.
\begin{algorithm}
\caption{Check Node Removal (CNR)}\label{alg:cnr}
\SetKwInOut{Input}{input}\SetKwInOut{Output}{output}
\Input{$\mathbf{H}$, $\mathtt{UNSAT} $, $\mathtt{t}$, $\mathtt{df}$}
\Output{Modified parity check matrix $\mathbf{H}_{FN}$}
$\mathbf{H}_{FN} = \mathbf{H}$\tcp*[r]{Initialization}
$\mathtt{rem}$ = []\tcp*[r]{Initialize possible leaf checks to remove}
\For{$\mathtt{check}$ in $\mathtt{UNSAT}$}{
    $\mathtt{leaf}$ $\gets$ All the leaf nodes of $T_t(\mathtt{check})$\;
    $\mathtt{rem}$ $\cup$ $\mathtt{leaf}$;
}
$\mathtt{dis}$ $\gets$ random.chose($\mathtt{rem}$,$\mathtt{df}$)\;
\ForEach{$c \in \mathtt{dis}$}{
    $\mathbf{H}_{FN} \gets \mathbf{H}_{FN}[\{1,\dots,m\} \setminus \{c\},\, :]$\tcp*[r]{Remove row indexed by check node $c$}
}
return $\mathbf{H_{FN}}$ \tcp*[r]{The modified matrix}
\end{algorithm}

\section{Find\_IMs algorithm}
\label{ap:findim}
In this appendix, we give the algorithm for the Find\_IMs function.
It consumes a set of unsatisfied stabilizer check nodes $\mathtt{UNSAT}$ as input and outputs a list of IM values for all the stabilizer check nodes of the underlying QLDPC codes.
The function also assumes a precomputed graph adjacency list of the neighborhood of the data and stabilizer checks $\mathtt{nv}$ and $\mathtt{nc}$ respectively.
\begin{algorithm}
\caption{Find\_IMs}\label{alg:find_im_opt}
\SetKwFunction{KwFn}{Find\_IMs}
\SetKwInOut{Input}{input}\SetKwInOut{Output}{output}

\Input{$\mathtt{UNSAT}$, $\mathtt{nc}$, $\mathtt{nv}$}
\Output{List of IM values for the stabilizer checks}

$\mathtt{im\_data} \gets [0]*n$ \tcp*[r]{n = number of data qubits}
$\mathtt{active\_vars} \gets \emptyset$

\BlankLine
\tcp{Step 1: IM values for data qubits}
\For{$\mathtt{cq} \in \mathtt{UNSAT}$}{
    \For{$\mathtt{dq} \in \mathtt{nc}[\mathtt{cq}]$}{
        $\mathtt{im\_data}[\mathtt{dq}] \gets \mathtt{im\_data}[\mathtt{dq}] + 1$\;
        $\mathtt{active\_vars} \gets \mathtt{active\_vars} \cup \{\mathtt{dq}\}$
    }
}

\BlankLine
\tcp{Step 2: IM values of stabilizer checks}
$\mathtt{im\_check} \gets [0] * m$ \tcp*[r]{m = number of stabilizers}

\For{$\mathtt{dq} \in \mathtt{active\_vars}$}{
    $\mathtt{val} \gets \mathtt{im\_data}[\mathtt{dq}]$\;
    \For{$\mathtt{cq} \in \mathtt{nv}[\mathtt{dq}]$}{
        $\mathtt{im\_check}[\mathtt{cq}] \gets \mathtt{im\_check}[\mathtt{cq}] + \mathtt{val}$
    }
}

\BlankLine
\Return $\mathtt{im\_check}$

\end{algorithm}
\section{Constant time algorithm to generate level $1$ computation tree}
Here we give the algorithm, which outputs a list of all the leaf check nodes from the level $1$ computation tree of an unsatisfied stabilizer check node.
In the QCCNR decoder, we exclusively use this algorithm to construct the sample space of check removal.
We discussed before that the removal of check nodes has always been done from level $1$ of any $T(c_r) : c_r \in \mathtt{UNSAT}$.
\begin{algorithm}
\caption{Layer1\_CT}\label{alg:layer1_tree}
\SetKwFunction{KwFn}{Generate\_Layer1}
\SetKwInOut{Input}{input}\SetKwInOut{Output}{output}

\Input{$\mathtt{check\_node}$, $\mathtt{nc}$, $\mathtt{nv}$}
\Output{List of leaf check nodes from the level $1$ computation tree $T(\mathtt{check\_node})$}

$\mathtt{neighbors} \gets \emptyset$

\BlankLine

\For{$\mathtt{dq} \in \mathtt{nc}[\mathtt{check\_node}]$}{
    \For{$\mathtt{cq} \in \mathtt{nv}[\mathtt{dq}]$}{
        \If{$\mathtt{cq} \neq \mathtt{check\_node}$}{
            $\mathtt{neighbors} \gets \mathtt{neighbors} \cup \{\mathtt{cq}\}$
        }
    }
}

\BlankLine
\Return $\text{list}(\mathtt{neighbors})$

\end{algorithm}

\section{Decoding modes of BP (min-sum algorithm)}
Here we describe the $\mathsf{MAIN\_BP}$ and $\mathsf{SUB\_BP}$ subroutines used in Algorithm \ref{alg:bp-qcnr}.
These are essentially the standard min-sum algorithm with a tracker for the change in the residual syndrome.
The tracker essentially indicates when the min-sum decoder gets trapped in the current settings and therefore different measures needs to be taken to mitigate the trapping set configurations.
\begin{algorithm}
\caption{MAIN\_BP / SUB\_BP}\label{alg:main_bp}
\SetKwInOut{Input}{input}\SetKwInOut{Output}{output}

\Input{Parity-check matrix $\mathtt{pcm}$, channel model $p$, maximum iterations $\mathtt{max_{iter}}$, $\mathtt{nc}$, $\mathtt{nv}$, tolerance $\mathtt{tol}$, syndrome $\mathtt{s}$}
\Output{Predicted error $\hat{e}_{bp}$ and predicted syndrome $\hat{s}_{bp}$}

$\mathtt{prev\_diff} \gets \text{None}$\;
$\mathtt{unchanged} \gets 0$\;

\BlankLine
\For{$t = 1$ \KwTo $\mathtt{max_{iter}}$}{
    \BlankLine
    \tcp{Single message passing iteration of min-sum algorithm}
    $\hat{\mathtt{e}}_{bp} \gets \text{min\_sum}(\mathtt{pcm}, p).\mathtt{decode}(\mathtt{s})$

    $\hat{\mathtt{s}}_{bp} \gets \mathbf{0}$\;

    \BlankLine
    \ForEach{$dq \in \mathrm{supp}(\hat{\mathtt{e}}_{bp})$}{
        \ForEach{$cq \in \mathtt{nv}[dq]$}{
            $\hat{\mathtt{s}}_{bp}[cq] \gets \hat{\mathtt{s}}_{bp}[cq] \oplus 1$
        }
    }

    $\mathtt{diff} \gets \|\mathtt{s}\| - \|\hat{s}_{bp}\|$\tcp*[r]{Weight difference}

    \BlankLine
    \If{$\mathtt{diff} = \mathtt{prev\_diff}$}{
        $\mathtt{unchanged} \gets \mathtt{unchanged} + 1$\;
        
        \If{$\mathtt{unchanged} = \mathtt{tol}$}{
            \textbf{break}\tcp*[r]{Early stopping}
        }
    }
    \Else{
        $\mathtt{unchanged} \gets 0$\;
    }

    $\mathtt{prev\_diff} \gets \mathtt{diff}$\;
}

\BlankLine
\Return $\hat{\mathtt{e}}_{bp}, \hat{\mathtt{s}}_{bp}$

\end{algorithm}
\newpage

\bibliography{references}

\end{document}